\documentclass[11pt,a4paper]{article}
\usepackage[colorlinks=true]{hyperref}
\hypersetup{urlcolor=blue, citecolor=red, linkcolor=blue}
\usepackage[table]{xcolor}
\usepackage{amsfonts}
\usepackage{amssymb}
\usepackage{amsmath}
\usepackage[utf8]{inputenc}
\usepackage{amsthm}
\usepackage[active]{srcltx}%(Para busqueda inversa)
\usepackage{mathtools}
\usepackage{graphicx}

\voffset = - 14mm
\newtheorem{theorem}{Theorem}[section]
\newtheorem{corollary}[theorem]{Corollary}
\newtheorem{lemma}[theorem]{Lemma}
\newtheorem{proposition}[theorem]{Proposition}

\theoremstyle{definition}
\newtheorem{definition}[theorem]{Definition}
\newtheorem{remark}[theorem]{Remark}
\newtheorem{example}[theorem]{Example}
\newcommand\xqed[1]{%
  \leavevmode\unskip\penalty9999 \hbox{}\nobreak\hfill
  \quad\hbox{#1}}
\newcommand\demo{\xqed{$\triangle$}}

\renewcommand{\d}{\mathrm{d}}

\title{{\sffamily Reduction and reconstruction of multisymplectic Lie systems}}

\author{{\sffamily 
$^a$Javier de Lucas%
\thanks{e-mail: javier.de.lucas@fuw.edu.pl \ \ 
ORCID: 0000-0001-8643-144X}\ ,\
$^b$Xavier Gr\`acia%
\thanks{e-mail: xavier.gracia@upc.edu \ \ 
ORCID: 0000-0003-1006-4086}\ ,\
$^c$Xavier Rivas%
\thanks{e-mail: xavier.rivas@unir.net \ \ 
ORCID: 0000-0002-4175-5157}\ ,
$^b$Narciso Rom\'{a}n-Roy%
\thanks{e-mail: narciso.roman@upc.edu \ \ 
ORCID: 0000-0003-3663-9861}\ ,\
$^d$Silvia Vilari\~no%
\thanks{e-mail: silviavf@unizar.es \ \ 
ORCID: 0000-0003-0404-1427}
}
\\[1ex]
\normalsize\itshape\sffamily 
$^a$Department of Mathematical Methods in Physics, University of Warsaw, 
Warszawa, Poland
\\[1pt]
\normalsize\itshape\sffamily 
$^b$Department of Mathematics,
Universitat Polit\`ecnica de Catalunya,
Barcelona, Spain
\\[1pt]
\normalsize\itshape\sffamily 
$^c$Escuela Superior de Ingenier\'{\i}a y Tecnolog\'{\i}a,
Universidad Internacional de La Rioja,
Logro\~no, Spain
\\[1pt]
\normalsize\itshape\sffamily 
$^d$Centro Universitario de la Defensa de Zaragoza $\&$ I.U.M.A., Zaragoza, Spain
}
\date{}%{\sffamily \date{}}}

\numberwithin{equation}{section}

\begin{document}

\maketitle

\leavevmode\vadjust{\kern -15mm}
%%%%%%%%%%%%%%%%%%%%%%%%%%%%%%%%%%%%%%%%%%%%%%%%%%%%%%%%%%%%%%%%
\begin{abstract}
A Lie system is a non-autonomous system of first-order ordinary differential equations describing the integral curves of a non-autonomous vector field taking values in a finite-dimensional real Lie algebra of vector fields, a so-called Vessiot--Guldberg Lie algebra.

In this work, multisymplectic forms are applied to the study of the reduction of Lie systems through their Lie symmetries. By using a momentum map, we perform a reduction and reconstruction procedure of multisymplectic Lie systems, which allows us to solve the original problem by analysing several simpler multisymplectic Lie systems. Conversely, we study how reduced multisymplectic Lie systems allow us to retrieve the form of the multisymplectic Lie system that gave rise to them. Our results are illustrated with examples from physics, mathematics, and control theory.
\end{abstract}


\noindent\textbf{MSC\,2020 codes:}   34A26, 34A05, 34A34, 17B66, 22E70.

%%%%%%%%%%%%%%%%%%%%%%%%%%%%%%%%%%%%%%%%%%%%%%%%%%%%%%%%%%%%%%%%
%%%%%%%%%%%%%%%%%%%%%%%%%%%%%%%%%%%%%%%%%%%%%%%%%%%%%%%%%%%%%%%%
% \newpage
{\setcounter{tocdepth}{2}
\def\baselinestretch{1}
\small
% hack per a eliminar l'espai vertical 1em
\def\addvspace#1{\vskip 1pt}
\parskip 0pt plus 0.1mm
\tableofcontents
}
%%%%%%%%%%%%%%%%%%%%%%%%%%%%%%%%%%%%%%%%%%%%%%%%%%%%%%%%%%%%%%%%

%%%%%%%%%%%%%%%%%%%%%%%%%%%%%%%%%%%%%%%%%%%%%%%%%%%%%%%%%%%%%%%%
\section{Introduction}\label{section:intro}
%%%%%%%%%%%%%%%%%%%%%%%%%%%%%%%%%%%%%%%%%%%%%%%%%%%%%%%%%%%%%%%%

A \textit{Lie system} is a system of first-order ordinary differential equations whose general solution can be expressed as an autonomous function depending on a generic finite family of particular solutions and a set of constants. The autonomous function is called a {\it superposition rule} of the Lie system. Standard examples of Lie systems are most types of Riccati equations \cite{GL13} and non-autonomous linear systems of ordinary differential equations \cite{Dissertationes,PW}.

The Lie--Scheffers theorem states that a Lie system amounts to a $t$-dependent vector field taking values in a finite-dimensional real Lie algebra of vector fields, a so-called \textit{Vessiot--Guldberg Lie algebra} \cite{Dissertationes}. The latter property gave rise to a number of methods for determining superposition rules \cite{CGM07,Dissertationes,PW}. Meanwhile, the Lie--Scheffers theorem also showed that being a Lie system is the exception rather than the rule \cite{Dissertationes}. Despite this, Lie systems admit numerous relevant physical and mathematical applications, as witnessed by the many works on the topic \cite{Am78,CCJL18,CGM00,Ru10,GMR97,Ib09,Ra06,Te60,PW}.

Recently, a lot of attention has been paid to Lie systems admitting a Vessiot--Guldberg Lie algebra of Hamiltonian vector fields and/or Lie symmetries relative to several types of geometric structures: Poisson \cite{CCJL18,CGM00,CLS13}, symplectic \cite{BBHLS13,BCHLS13,CCJL18,CLS13,Ru10}, Dirac \cite{CCJL18,CGLS14}, $k$-symplectic \cite{LV15}, multisymplectic \cite{GLMV19}, Jacobi \cite{HLS15}, Riemann \cite{HLT17}, and others \cite{CCJL18,LL18}. Surprisingly, this led to finding much more applications of Lie systems than in the literature dealing with mere Lie systems \cite{BHLS15,CCJL18,LL18}. Such structures allow for the construction of superposition rules, constants of motion, and other properties of Lie systems in an algebraic manner without relying on the solving of complicated systems of partial or ordinary differential equations \cite{CGM00,CGM07,Dissertationes,PW}. Geometric structures also explain the geometric meaning of superposition rules \cite{BCHLS13} and gave rise to the study of more general differential equations \cite{BCFHL18} as well as physical and mathematical problems \cite{CCJL18,LL18}.

This work pioneers the analysis of the reduction of Lie systems with compatible geometric structures \cite{LS20}, in particular with multisymplectic structures \cite{GLMV19}. The reduction of multisymplectic Lie systems allows us to simplify the system by reducing its number of variables while ensuring the existence of a compatible multisymplectic structure for the reduced system, which is again a multisymplectic Lie system. This may be used, for instance, to obtain superposition rules for the reduced multisymplectic Lie systems, which in turn can be derived through the multisymplectic structure \cite{GLMV19}. It should be noted that our paper presents for the first time the use of a type of reduction procedure for Lie systems (cf.\ \cite{Dissertationes,LS20}). 

More particularly, by proposing a certain type of reduction process of multisymplectic structures, which is interesting on its own due to the very wide interest in any kind of multisymplectic reduction process (see \cite{Bl21,BMR22,EMN18} and references therein), we perform a reduction process of the so-called multisymplectic Lie systems introduced in \cite{GLMV19}. More specifically, our reduction process is focused on Lie systems of locally automorphic type, i.e. they can be considered locally as Lie systems on Lie groups of a very particular but relevant type \cite{GLMV19}.   
Finally, we also describe a reconstruction procedure to recover the initial $t$-dependent multisymplectic Lie system from several of its multisymplectic reductions. It is worth noting that our techniques are illustrated by the study of relevant physical and mathematical systems that appear in the study of quantum mechanical systems, control systems, and other topics. 

The structure of the paper goes as follows. Section \ref{section:basic_concepts} presents the basic theory of Lie systems, including the Lie--Scheffers theorem, and some results on automorphic Lie systems. It also offers a review on multisymplectic structures and the notion of multisymplectic Lie system \cite{GLMV19} is presented. The last part of this section is devoted to a survey of the notions of unimodular Lie algebras and unimodular Lie groups.

In Section \ref{section:LS-IF} we review some results on locally automorphic Lie systems and give a procedure to find some of their invariants. We also work out two examples to better illustrate these concepts. The first is the so-called generalised Darboux--Brioschi--Halphen system, which appears in the study of triply orthogonal surfaces and vacuum Einstein equations for hyper-K\"ahler Bianchi-IX metrics \cite{CH03,Darboux,Halphen}. Additionally, it also occurs when reducing the self-dual Yang--Mills equations corresponding to an infinite-dimensional gauge group of diffeomorphisms of a
three-dimensional sphere \cite{CH03}. The second example is a first-order control system on $\mathbb{R}^5$ with two controls \cite{Ni00,Ra06}.

Section \ref{section:reduction} devises a new reduction procedure for multisymplectic Lie systems. We begin by working out a couple of introductory examples: the Schwarz equation and dissipative quantum harmonic oscillators. Having studied some examples, we develop a multisymplectic Lie systems reduction theory. These results are applied to the control system introduced in the previous section and to quantum harmonic oscillators with a spin-magnetic term.

A reconstruction procedure for multisymplectic Lie systems is developed in Section \ref{section:reconstruction}. This reconstruction is achieved by combining several different reductions in an appropriate way and allows us to recover the initial multisymplectic Lie system. This procedure is applied to the example of quantum harmonic oscillators with a spin-magnetic term introduced in the previous section.

%%%%%%%%%%%%%%%%%%%%%%%%%%%%%%%%%%%%%%%%%%%%%%%%%%%%%%%%%%%%%%%%
\section{Some basic concepts and notations}
\label{section:basic_concepts}
%%%%%%%%%%%%%%%%%%%%%%%%%%%%%%%%%%%%%%%%%%%%%%%%%%%%%%%%%%%%%%%%

Let us assume some general statements to hold throughout the work unless otherwise explicitly stated.
All mathematical objects are smooth, real, and globally defined.
This allows us to avoid non-relevant technical problems while stressing the main ideas of our theory. Hereafter, $N$ will represent an $n$-dimensional connected manifold. All manifolds are considered connected. The sum over crossed repeated indices is understood.

%%%%%%%%%%%%%%%%%%%%%%%%%%%%%%%%%%%%%%%%%%%%%%%%%%%%%%%%%%%%%%%%
\subsection{Generalised distributions and \texorpdfstring{$t$}--dependent vector fields}
\label{subsection:generalised_distributions}
%%%%%%%%%%%%%%%%%%%%%%%%%%%%%%%%%%%%%%%%%%%%%%%%%%%%%%%%%%%%%%%%

Consider a Lie algebra $V$ with a Lie bracket $[\cdot,\cdot]$. Given two subsets $\mathcal{A}, \mathcal{B} \subset V$,  we denote by $[\mathcal{A},\mathcal{B}]$ the linear space generated by the Lie brackets between elements of $\mathcal{A}$ and $\mathcal{B}$. Meanwhile, ${\rm Lie}(\mathcal{B})$ stands for the smallest Lie subalgebra of $V$ containing $\mathcal{B}$.

A {\it Stefan--Sussmann} (or {\it generalised}) {\it distribution}
on a
manifold $N$ is a correspondence $\mathcal{D}$  associating each $x\in N$ with a linear
subspace $\mathcal{D}_x\subset {\rm T}_xN$. 
The dimension of $\mathcal{D}_x$ is called the {\it rank} of $\mathcal{D}$ at~$x$.
It is said that $\mathcal{D}$ is {\it regular at} $x\in N$ when
its rank is constant on a neighbourhood of~$x$. 
Meanwhile, the generalised distribution
$\mathcal{D}$ is called \emph{regular} when its rank is constant on the whole~$N$. Regular generalised distributions are called \emph{distributions}. A vector field $Y\in \mathfrak{X}(N)$ 
takes values in~$\mathcal{D}$, let us write it $Y\in\mathcal{D}$, when
$Y_x\in\mathcal{D}_x$ for all $x\in N$. 

We write $\mathfrak{X}(N)$
for the $C^\infty(N)$-module of vector fields on~$N$. Meanwhile, $\Omega(N)$ and $\Omega^k(N)$ stand for the $C^\infty(N)$-modules of differential forms and differential $k$-forms on~$N$, respectively.
A set $\mathcal{V}\subset \mathfrak{X}(N)$ of vector fields on~$N$ allows us to define a generalised distribution $\mathcal{D}^\mathcal{V}$ on $N$ mapping each $x\in N$ to the linear span of all the values of its vector fields at $x$, i.e. 
$\mathcal{D}^\mathcal{V}_x = \mathrm{span}\{X_x \mid X \in \mathcal{V} \}$.

A {\it $t$-dependent vector field} on $N$ is a mapping
$X \colon (t,x) \in \mathbb{R} \times N  \mapsto  X(t,x) \in {\rm T}N$
so that $\tau_N\circ X=\pi_2$,
where
$\pi_2 \colon (t,x) \in \mathbb{R} \times N  \mapsto  x \in N$ and $\tau_N:{\rm T}N\rightarrow N$ is the canonical tangent bundle projection. A $t$-dependent vector field~$X$ on $N$
amounts to  
a $t$-parameterised family of vector fields
$\{X_t\}_{t\in\mathbb{R}}$ on~$N$,
with 
$X_t \colon x \in N  \mapsto  X(t,x) \in {\rm T}N$
\cite{Dissertationes}. This enables us to relate $t$-dependent vector fields to the following geometric structures.

An {\it integral curve} of $X$ 
is a curve $\gamma \colon \mathbb{R} \to N$
satisfying that 
\begin{equation}\label{Eq:Sys}
\frac{{\rm d}\gamma}{{\rm d} t}(t) = X(t,\gamma(t))
\,,
\quad 
\forall t\in \mathbb{R}
\,.
\end{equation}
Hence,
$\widetilde\gamma:t\in \mathbb{R}\mapsto (t,\gamma(t))\in \mathbb{R}\times N$
becomes an integral curve
of the {\it autonomisation} $\widetilde X$ of~$X$,
namely the vector field 
$\displaystyle
\widetilde X = \partial/\partial t + X$ 
on 
$\mathbb{R}\times N$ 
\cite{FM,Dissertationes}.
On the contrary,
if $\widetilde\gamma \colon \mathbb{R} \to \mathbb{R} \times N$ 
is an integral curve of the autonomisation~$\widetilde X$
and a section of the bundle 
$\pi_1 \colon (t,x) \in \mathbb{R} \times N  \mapsto  t \in \mathbb{R}$,
then
$\gamma = \pi_2 \circ \widetilde\gamma$
is an integral curve of~$X$. Note that $\widetilde{\gamma}$ is a section of the bundle $\pi_1$ if, and only if, $\pi_1\circ\widetilde{\gamma}={\rm Id}_\mathbb{R}$. It is worth noting, for clarity, that  
$\sigma(t)=(t+1,\gamma(t+1))$, where $\gamma(t)$ is a solution to (\ref{Eq:Sys}),  is not a section of $\pi_1$, but it is an integral curve of $\widetilde{X}$.

Every $t$-dependent vector field $X$ on a manifold $N$ gives rise to a unique system (\ref{Eq:Sys}) on $N$ describing its integral curves. 
Conversely, a system (\ref{Eq:Sys}) describes the integral curves $\widetilde\gamma:t\in \mathbb{R}\mapsto (t,\gamma(t))\in \mathbb{R}\times N$ of the autonomisation of a unique $t$-dependent vector field $X$ on~$N$. 
This allows us to use $X$ to identify both a $t$-dependent vector field on $N$ and its associated system (\ref{Eq:Sys}). This identification will simplify the terminology of our paper without leading to any misunderstanding.

\begin{definition}
\label{def:smallest_Lie_algebra}
The {\it smallest Lie algebra} of a $t$-dependent vector field~$X$ on~$N$
is the smallest (in the sense of inclusion) Lie algebra over the reals, $V^X$, 
including the vector fields $\{X_t\}_{t\in\mathbb{R}}$, namely
$V^X={\rm Lie}(\{X_t\}_{t\in\mathbb{R}})$. 
The {\it associated distribution} of~$X$, denoted by $\mathcal{D}^{V^X}$, is the generalised distribution on~$N$ spanned by the elements of the Lie algebra of vector
fields $V^X$. 
\end{definition}

It can only be ensured that the rank of $\mathcal{D}^{V^X}$
is
constant on the connected components of an open and dense subset of~$N$,
where $\mathcal{D}^{V^X}$ becomes regular, involutive, and integrable (see \cite{CLS13}).
The most relevant instance for us is when $V^X$ is a finite-dimensional Lie algebra and, therefore, the generalised distribution $\mathcal{D}^{V^X}$ becomes integrable on the whole~$N$ in the sense of Stefan--Sussmann 
(see \cite[p.\,63]{JPOT} for details).

%%%%%%%%%%%%%%%%%%%%%%%%%%%%%%%%%%%%%%%%%%%%%%%%%%%%%%%%%%%%%%%%
\subsection{Lie systems}\protect\label{subsection:Lie systems}
%%%%%%%%%%%%%%%%%%%%%%%%%%%%%%%%%%%%%%%%%%%%%%%%%%%%%%%%%%%%%%%%

Let us survey some fundamental notions in the theory of Lie systems to be used hereafter (see \cite{Dissertationes} for details).

\begin{definition}\label{def:superposition_rule} A {\it superposition rule} depending on $m$ particular solutions for a system~$X$ in $N$ 
is a function
$\Phi \colon N^{m} \times N \rightarrow N$,
$x=\Phi(x_{(1)}, \ldots,x_{(m)};\lambda)$,
so that the general solution, $x(t)$, of $X$ can be written as
$x(t)=\Phi(x_{(1)}(t), \ldots,x_{(m)}(t);\lambda),$
where $x_{(1)}(t),\ldots,x_{(m)}(t)$ is any generic family of
particular solutions of $X$ and $\lambda$ is a point in $N$ that is related to the initial conditions. 
A {\it Lie system} is a non-autonomous system of first-order ordinary differential equations that admits a superposition rule.
\end{definition}

The {\it Lie--Scheffers theorem} \cite{CGM07,LS} states the conditions that ensure that a system $X$ possesses a superposition rule:

\begin{theorem}
\label{th:characterization_superposition_rule}
A system $X$ on $N$ possesses a superposition rule if, and only if,
$X = {{\sum_{\alpha=1}^r}} b_\alpha(t) X_\alpha$,
for a
family $X_1,\ldots,X_r$ of vector fields on $N$ spanning
an $r$-dimensional Lie algebra and a set $b_1(t),\ldots,b_r(t)$ of $t$-dependent functions.
\end{theorem}

In other words,
$X$ possesses a superposition rule if, and only if,
 $\dim V^X<\infty$. 
Then, a Lie system can always be written as 
$X=\sum_{\alpha=1}^rb_\alpha(t)X_\alpha$, where $X_1,\ldots,X_r$ span a Lie algebra $V$ that may strictly contain $V^X$. 
Then, $V$ is called a
{\it Vessiot--Guldberg Lie algebra} of~$X$. It should be noted that $V$ may not be univocally defined for a system $X$, while $V^X$ is always characterised by $X$  \cite{Dissertationes}.

Previous comments allow us to hereafter denote a \emph{Lie system} as a triple 
$(N,X,V)$,
where $N$ is a manifold and $V$ is a Vessiot--Guldberg Lie algebra of a $t$-dependent vector field,~$X$, on~$N$.

Let us show how the integration of a Lie system on a manifold 
can be reduced to integrate one of its Vessiot-Guldberg Lie algebras to a Lie group action and to know a particular solution of a Lie system on a Lie group of a very specific type, a so-called automorphic Lie system \cite{CGM00,Ve1904}. With this aim, let us recall some basic results regarding Lie group actions.

Let 
$\varphi \colon D\subset G \times N \to N$ be a local (left) Lie group action
and define
$\varphi_{x}:g\in D_G\subset G\mapsto \varphi(g,x)\in N$ so that $D_G=\{g\in G\,|\, (g,x)\in D\}$ \cite{Palais}.
For every $\xi \in \mathrm{T}_eG$,
we set the fundamental vector field $\xi_N \in \mathfrak{X}(N)$ to be
$\xi_N(x) ={\rm T}_e \varphi_{x} (\xi)$. Let us denote by $\xi^L$ and $\xi^R$ 
the respective invariant vector fields associated with 
$\xi \in \mathrm{T}_eG$.
Hence, there are linear isomorphisms $\xi\in \mathrm{T}_eG\mapsto \xi^L\in \mathfrak{X}_L(G)$ and  $\xi\in \mathrm{T}_eG\mapsto \xi^R\in \mathfrak{X}_R(G)$, where $\mathfrak{X}_L(G)$ and $\mathfrak{X}_R(G)$ are the linear spaces of left- and right- invariant vector fields on~$G$.
The space
$\mathfrak{g} = \mathrm{T}_eG$  
inherits a Lie algebra structure from $\mathfrak{X}_L(G)$ via the identification given by the  linear isomorphism $\xi\in \mathrm{T}_eG\mapsto \xi^L\in \mathfrak{X}_L(G)$.
Then
(cf.\ \cite[Ch.\,20]{Lee}),
\begin{enumerate}
\itemsep 0pt
\item 
The map 
$\widehat\varphi \colon \mathfrak{g} \to \mathfrak{X}(N)$,
$\xi \mapsto \xi_N$,
is a Lie algebra antihomomorphism
(the \emph{infinitesimal generator} of~$\varphi$).
\item
For every $x \in N$,
the \emph{right}-invariant vector field $\xi^R$
and the fundamental vector field $\xi_N$ 
are $\varphi_x$-related \cite[p.\,182]{Lee}, namely $T_g\varphi_x(\xi^R(g))=\xi_N(\varphi_x(g))$ for every $g\in G$.
\end{enumerate}

In general, an antihomomorphism 
$\mathfrak{g} \to \mathfrak{X}(N)$ 
is called a (left) Lie algebra action, and,
when $\mathfrak{g}$ is finite-dimensional,
it can be integrated to a Lie group action.
More precisely:

\begin{theorem}\label{Th:Integration}
Let $\mathfrak{g}$ be a finite-dimensional Lie algebra,
and let $G$ be a Lie group with Lie algebra $\mathfrak{g}$.
Given a Lie algebra morphism $\widehat\varphi:\mathfrak{g}\rightarrow \mathfrak{X}(N)$,
there exists a local Lie group action 
$\varphi:G\times N\rightarrow N$
whose infinitesimal generator is~$\widehat\varphi$.
If $G$ is connected and simply connected and the vector fields $\widehat\varphi(\mathfrak{g})$ are complete,
then $\varphi$ becomes a global Lie group action. 
\end{theorem}
The proof of these results, and other related facts, can be found, for instance, in
\cite[p.\,58]{Palais},
\cite[p.\,207]{Bo68}.
This Lie group action is the device that relates the Lie system on~$N$ to a Lie system on~$G$:

\begin{theorem}
\label{Th:XG-X}
Let $(N,X,V)$ be a Lie system given by
$X = {{\sum_{\alpha=1}^r}} b_\alpha(t) X_\alpha$,
where
$X_1,\ldots,X_r$ is a basis of the Vessiot--Guldberg
Lie algebra~$V$ and $b_1(t),\ldots,b_r(t)$ are arbitrary $t$-dependent functions.
Let $G$ be a Lie group whose Lie algebra is isomorphic to $V$,
and let $\varphi:D_G\in G\times N\rightarrow N$ be a local Lie group action 
as given by Theorem \ref{Th:Integration}.
If $X_\alpha^R$ is the right-invariant vector field on~$G$
related to the vector field $X_\alpha$ through $\varphi$, with $\alpha=1,\ldots,r$,
then
\begin{enumerate}
\item 
The triple 
$(G,X^G,V^G)$,
where
\begin{equation}
\label{Eq:EquLie}
X^G(t,g) = \sum_{\alpha=1}^r b_\alpha(t) X_\alpha^R(g),\qquad \forall g\in G,\quad \forall t\in \mathbb{R},
\end{equation}
and
$V^G = \mathfrak{X}_R(G)$,
is a Lie system on~$G$.
\item
For every $x_0 \in N$ and $t\in \mathbb{R}$,
the vector field
$X_t^G$ is $\varphi_{x_0}$-related with~$X_t$; 
namely, 
the $t$-dependent vector fields $X^G$ and $X$ are $\varphi_{x_0}$-related.
\item
If $g(t)$ is the integral curve of $X^G$ with $g(0)=e$,
and $x_0 \in N$,
then 
$x(t) = \varphi(g(t),x_0)$
is the integral curve of~$X$ with $x(0)=x_0$.
\end{enumerate}
\end{theorem}
The proof of this result is almost immediate:
the vector fields $X_\alpha^R$ span the Lie algebra
$\mathfrak{X}_R(G)$, the $t$-dependent vector field
$X^G$ is $\varphi_{x_0}$-related with~$X$,
and this implies that 
$\varphi_{x_0}$ maps integral curves of $X^G$ to integral curves of~$X$ (see
\cite{CGM00,Dissertationes}
for details).

Hence, knowing the explicit form of $\varphi$ allows us to obtain the solutions to~$X$
% the explicit integration of the Lie system $X$ 
via a single particular solution to~(\ref{Eq:EquLie}).
The other way around, the general solution of $X$ determines  
the integral curve of (\ref{Eq:EquLie}) with $g(0)=e$ 
using an algebraic system of equations determined by~$\varphi$. 

Lie systems (\ref{Eq:EquLie}) are frequently called 
\emph{automorphic Lie systems} 
\cite{BS08,Ve1904}. 
Due to their specific structure, automorphic Lie systems possess invariant differential forms with respect to their evolution (see \cite{GLMV19} for details). 
From now on, we say that the automorphic Lie system 
$(G,X^G,\mathfrak{X}_R(G))$ is the automorphic Lie system
 \emph{related} to $(N,X,V)$.

The automorphic Lie system related to a $t$-dependent right-invariant vector field (\ref{Eq:EquLie})
is invariant relative to right multiplications 
$R_g: g'\in G \rightarrow g'g\in G$. 
Then, if $g_1(t)$ is the particular solution to the system $X^G$ with an initial condition $g_1(0)\in G$, then $R_gg_1(t)$ is the particular solution of $X^G$ with the initial condition $g_1(0)g$. 
Consequently, the general solution to $X^G$ can be brought into the form  
$g(t)=\Phi^G(g_1(t);g) := R_g \,g_1(t)$, 
which leads to the definition of the superposition rule 
$\Phi^G \colon (g_1;g) \in G\times G \mapsto R_gg_1 \in G$, which  
depends on a unique particular solution \cite{Dissertationes}.

%%%%%%%%%%%%%%%%%%%%%%%%%%%%%%%%%%%%%%%%%%%%%%%%%%%%%%%%%%%%%%%%
\subsection{Multisymplectic manifolds and Lie systems}\label{subsection:multisymplectic}
%%%%%%%%%%%%%%%%%%%%%%%%%%%%%%%%%%%%%%%%%%%%%%%%%%%%%%%%%%%%%%%%

This section reviews the fundamental properties of multisymplectic manifolds to be used hereafter
(see \cite{CIL-96a,CIL-96b,EIMR-2012} for details). 

A differential $k$-form $\omega$ on~$N$
is called \emph{one-nondegenerate}
if, and only if, the vector bundle morphism
\[
\begin{array}{rrcl}
\omega^\flat \colon & \mathrm{T}N & \to & \mathsf{\Lambda}^{k-1} \mathrm{T}^*N
\\
& X_p & \mapsto & \iota_{X_p} \omega_p,
\end{array}
\]
where $X_p\in \mathrm{T}_pN$ and $\iota_{X_p}\omega_p$ stands for an inner contraction, is injective.
If this is the case,
the induced morphism of $\mathcal{C}^\infty(N)$-modules
$\widehat\omega\colon\mathfrak{X}(N) \to \Omega^{k-1}(N)$,
$X  \mapsto  \iota_X \omega$,
is also injective.

\begin{definition}\label{def:Mult}
A \emph{multisymplectic $k$-form} on an $n$-dimensional manifold $N$ is a $1$-nondegenerate and closed differential 
$k$-form $\Theta \in \Omega^k(N)$.
A \emph{multisymplectic manifold} of degree $k$ is a pair $(N,\Theta)$,
where $\Theta$ is a multisymplectic $k$-form on~$N$.
\end{definition}

The multisymplectic two-forms are the symplectic forms.
Multisymplectic $n$-forms on $N$ are the volume forms on~$N$.
We hereafter assume that $\dim N\geq 2$ and, thus, every multisymplectic $k$-form has degree $k \geq 2$.
The `multisymplectic' term can be misleading and it should not be confused with other structures, such as $k$-symplectic and polysymplectic \cite[p.\,21]{dLeon2015}, where several differential forms appear. Anyhow, we will stick to the standard nomenclature, which has been established for almost fifty years now \cite{K73,KT79}.

\begin{definition}\label{def:multi_hamvf}
Let $(N,\Theta)$ be a multisymplectic manifold of degree~$k$.
A vector field $X$ on $N$ is \textit{locally Hamiltonian} if $\iota_X\Theta$ is closed.
A vector field $X$ is \textit{Hamiltonian} if $\iota_X\Theta$ is exact, i.e. $\iota_X\Theta=d\Upsilon$ for a differential $(k-2)$-form $\Upsilon_X$ on~$N$.
In previous cases, $\iota_X\Theta$ and $\Upsilon_X$ are called  \textit{Hamiltonian} $(k-1)$- and $(k-2)$-forms associated with~$X$, respectively.
\end{definition}

A vector field is locally Hamiltonian if, and only if, $\Theta$ is invariant by~$X$, that is,
\[
\mathcal{L}_X \Theta = 0\,,
\]
where $\mathcal{L}_X\Theta$ stands for the  Lie derivative of $\Theta$ relative to $X$. 

To keep the notation simple, we will frequently call $\iota_X\Omega$ and $\Upsilon_X$ Hamiltonian forms associated with $X$. Note that the degree is obvious from the context and does not need to be detailed. 

In analogy with symplectic geometry, it is possible to define certain brackets on the spaces of Hamiltonian forms on multisymplectic manifolds. Notwithstanding, their properties generally differ significantly from the properties of the brackets defined in symplectic geometry.

\begin{definition}\label{def:AH}
Let $(N,\Theta)$ be a multisymplectic manifold of degree~$k$.
\begin{itemize}
\item
A {\it bracket between Hamiltonian $(k-1)$-forms} can be set as follows. Let
$\xi,\zeta \in {\rm Im}\,\widehat{\Theta} \subset \Omega^{k-1}(N)$, 
and let $X,Y\in\mathfrak{X}(N)$ be the unique vector fields such that $\iota_{X}\Theta=\xi$ and $\iota_{Y}\Theta=\zeta$.
The {\sl bracket between $\xi$ and $\zeta$} is defined by
\begin{equation}
\{\xi,\zeta\} = \iota_{[Y,X]}\Theta\ \in {\rm Im}\,\widehat{\Theta}
\,.
\label{braform}
\end{equation}
This bracket satisfies the Jacobi identity and, therefore, becomes a Lie bracket.

\item
The {\sl bracket between Hamiltonian $(k-2)$-forms} will be now defined.
Let $X,Y$ be Hamiltonian vector fields and let 
$\Upsilon_X,\Upsilon_Y \in \Omega^{k-2}(N)$ 
be Hamiltonian forms of $X$ and $Y$, respectively.  
Then, 
$$
\{\Upsilon_X,\Upsilon_Y\} = \iota_Y\iota_X\Theta \,
$$
gives rise to a bracket on $(k-2)$-forms. It can be proved that the bracket of Hamiltonian $(k-2)$-forms needs not be a Lie bracket for $k>2$ \cite{CIL-96a}.
\end{itemize}
\end{definition}

Although the brackets for Hamiltonian $(k-1)$ and $(k-2)$-forms have been denoted in the same way, 
this will not lead to confusion and will simplify the notation. 

Note that
\begin{equation}
\d\{\Upsilon_X,\Upsilon_Y\}=
\d\iota_Y\iota_X\Theta=
\iota_{[Y,X]}\Theta=
\{\d\Upsilon_X,\d\Upsilon_Y\} \,.
\label{relation1}
\end{equation}
The equality
$\d\{\Upsilon_X,\Upsilon_Y\}=\iota_{[Y,X]}\Theta$, shows that $[Y,X]$ is a Hamiltonian vector field
admitting a Hamiltonian form $\{\Upsilon_X,\Upsilon_Y\}$. 
Consequently, the space of (locally) Hamiltonian vector fields (${\rm Ham}_{\mathrm{loc}}(N)$) ${\rm Ham}(N)$ 
becomes a Lie algebra, 
and the map attaching a (locally) Hamiltonian vector field $X$ to $\iota_{X}\Theta$ is an injective Lie algebra anti-homomorphism. 

\medskip

Let us recall the notion of a multivector field, 
which is needed in our reduction procedure for multisymplectic Lie systems
(see \cite{EMR98,FR} for more details).
%\begin{definition}
%\label{def:multivector}
An \textit{$\ell$-multivector field} on $N$ is a section of the bundle  $\mathsf{\Lambda}^\ell(\mathrm{T}N)$.
An $\ell$-multivector field $Y$ on $N$ is said to be \emph{decomposable} if there exists a family of vector fields
$Y_1,\ldots, Y_\ell \in \mathfrak{X}(N)$ 
such that 
$Y = Y_1 \wedge \ldots \wedge Y_\ell$.
We will denote by $\mathfrak{X}^\ell(N)$
the set of $\ell$-multivector fields.

Let $(N,\Theta)$ be a multisymplectic manifold of degree~$k$.
An $\ell$-multivector field $Y$ is \textit{Hamiltonian} 
(with respect to~$\Theta$)
if there exists a $(k-\ell-1)$-form $\theta$ such that 
$\iota_Y \Theta = \d\theta$. This notion allows us to generalise the notion of Hamiltonian vector fields to multivector fields. 
Furthermore, $Y$ is \textit{locally Hamiltonian} or \textit{multisymplectic} if 
$\mathcal{L}_Y \Theta = 0$ (see \cite{CIL-96a,CIL-96b}).

Finally, let us recall the notion of a multisymplectic Lie system \cite{GLMV19}.
\begin{definition}
    A (locally) multisymplectic Lie system is a triple $(N,\Theta,X)$, where $X$ is a Lie system whose smallest Lie algebra $V^X$ is a finite-dimensional real Lie algebra of (locally) Hamiltonian vector fields relative to a multisymplectic structure $\Theta$ on~$N$.
\end{definition}

%%%%%%%%%%%%%%%%%%%%%%%%%%%%%%%%%%%%%%%%%%%%%%%%%%%%%%%%%%%%%%%%
\subsection{Unimodular Lie algebras}
\label{subsection:unimodular}
%%%%%%%%%%%%%%%%%%%%%%%%%%%%%%%%%%%%%%%%%%%%%%%%%%%%%%%%%%%%%%%%

This section surveys the fundamentals of unimodular Lie algebras and Lie groups to be used hereafter. 

Let $G$ be a Lie group with Lie algebra $\mathfrak{g}=\mathrm{T}_eG$. 
A (left) {\it Haar measure} on~$G$ is a left-invariant volume form on~$G$ \cite{Ja98}. 
Every Lie group admits a Haar measure given by a left-invariant volume form, and it is unique up to a non-zero multiplicative constant 
(cf.\ \cite{Bu13}).

Let $\{X^L_1,\ldots, X_r^L\}$ be a basis of~the Lie algebra $\mathfrak{X}_L(G)$ of left-invariant vector fields on~$G$ and
let
$\{\eta^L_1,\ldots,\eta_r^L\}$ 
be its dual basis of left-invariant differential one-forms.
Then, any left-invariant volume form on~$G$ is a non-zero scalar multiple of
$$
\Theta = \eta_1^L \wedge\ldots\wedge \eta_r^L
\,.
$$
If $X^L$ is a left-invariant vector field on~$G$, then \cite{GLMV19}:
\begin{equation}
\label{Eq:exp}
\mathcal{L}_{X^L} \Theta = -{\rm Tr}({\rm ad}_{X^L}) \,\Theta
\,;
\end{equation}
where ${\rm Tr}$ stands for the trace of an endomorphism,
and 
${\rm ad}\colon v\in \mathfrak{g} \mapsto  {\rm ad}_v\in {\rm End}(\mathfrak{g})$, 
is the adjoint representation of a Lie algebra~$\mathfrak{g}$ and
given by 
${\rm ad}_v w = [v,w]$.

A Lie group is called {\it unimodular} if its Haar measure is {\it bi-invariant}, namely it is left- and right-invariant \cite{Mi76}. For instance, Abelian,  compact, and semi-simple Lie groups, are all unimodular \cite{Yo15}. This work focusses on the Lie algebras of unimodular Lie groups, whose most important properties (for our purposes) are detailed in the definition and proposition below (for details, see \cite{GLMV19}).

\begin{definition}
\label{def:unimodular_lie_algebra}
A finite-dimensional Lie algebra $\mathfrak{g}$ is called {\it unimodular} when the maps 
${\rm ad}_v \in {\rm End}(\mathfrak{g})$
are traceless
---we say then that the adjoint representation is {\it traceless}.
\end{definition}

\begin{proposition}
\label{prop:unimodular_LG}
A (connected) Lie group $G$ is unimodular if, and only if, its Lie algebra is unimodular.
\end{proposition}

\begin{remark}
\label{Rem:Important}
It is worth noting that each
left-invariant vector field $X$ on~$G$ has a flow of the form $\phi:t\in\mathbb{R}\mapsto R_{\exp(t X_e)}\in {\rm Diff}(G)$. Then, 
a vector field $Y$ on a connected Lie group~$G$ is right-invariant 
if, and only if, it commutes with \emph{every} left-invariant vector field~$X^L$, namely
$\mathcal{L}_{X^L} Y = 0$.
This also applies to tensor fields on~$G$.
\end{remark}

\begin{definition}\label{def:grassman}
Let  $\mathfrak{g}$ be an $n$-dimensional  Lie algebra. We write $\mathcal{G}(r,\mathfrak{g})$ for the set of decomposable $r$-multivectors spanning $r$-dimensional unimodular Lie subalgebras. The set $\mathcal{G}(\mathfrak{g})=\bigcup_{r=0}^n\mathcal{G}(r,\mathfrak{g})$ is called the \textit{unimodular Grassmannian} of~$\mathfrak{g}$.
\end{definition}

%%%%%%%%%%%%%%%%%%%%%%%%%%%%%%%%%%%%%%%%%%%%%%%%%%%%%%%%%%%%%%%%
\section{Locally automorphic Lie systems and invariant forms}
\label{section:LS-IF}
%%%%%%%%%%%%%%%%%%%%%%%%%%%%%%%%%%%%%%%%%%%%%%%%%%%%%%%%%%%%%%%%

In this section, we recall conditions 
that ensure the existence of a multisymplectic form~$\Theta$ 
invariant with respect to the elements of a Vessiot--Guldberg Lie algebra~$V$.

In general, it is difficult to find multisymplectic forms compatible with a Lie system $X$ admitting a Vessiot--Guldberg Lie algebra~$V$, because this requires finding adequate solutions $\Theta$
to a system of partial differential equations 
$\mathcal{L}_Y\Theta=0$ for every $Y\in V$. 
However, we can develop several simpler methods to find compatible invariant forms 
for a special class of Lie systems with many relevant physical applications: 
the so-called locally automorphic Lie systems.

%%%%%%%%%%%%%%%%%%%%%%%%%%%%%%%%%%%%%%%%%%%%%%%%%%%%%%%%%%%%%%%%
\subsection{Locally automorphic Lie systems}
%%%%%%%%%%%%%%%%%%%%%%%%%%%%%%%%%%%%%%%%%%%%%%%%%%%%%%%%%%%%%%%%

\begin{definition} 
A {\it locally automorphic Lie system} is a triple $(N,X,V)$, 
where $X$ is a Lie system on the manifold $N$ with a Vessiot--Guldberg Lie algebra~$V$ 
such that $\dim V=\dim N$ and $\mathcal{D}^V={\rm T}N$.
\end{definition}

The following result, whose proof can be found in \cite{GLMV19}, shows that locally automorphic Lie systems are locally diffeomorphic to automorphic Lie systems, thus motivating the previous definition. 

\begin{theorem}
\label{Trivial}
Consider a locally automorphic Lie system $(N,X,V)$. Let $G$ be a Lie group whose Lie algebra is isomorphic to~$V$,
let $\varphi$ be a local action of~$G$ on~$N$
obtained from the integration of~$V$,
and let $(G,X^G,V^G)$ be the corresponding automorphic Lie system on~$G$ given by Theorem \ref{Th:XG-X}.
Then, for every $x \in N$, the map 
$\varphi_x = \varphi(\cdot,x)$
is a local diffeomorphism such that ${\rm T}_g\varphi_{x}(X^G_t(g))=X_t(\varphi_x(g))$ for every $t\in \mathbb{R}$.
\end{theorem}

Recall that the action $\varphi$ can be ensured to be globally defined if, and only if, $G$ is simply connected and the Lie algebra $V$ consists of complete vector fields \cite{Palais}. From now on, we assume that the action $\varphi$ is defined globally. 

The mapping $\varphi$ allows us not only to find local diffeomorphisms $\varphi_x:G\rightarrow N$, with $x\in N$, but also maps certain geometric structures related to the locally automorphic Lie system $(N,X,V)$ with the associated automorphic Lie system $(G,X^R,V^R)$. 

As a consequence of Theorem \ref{Trivial}, we have the following corollaries.

\begin{corollary}
\label{cor:superposition_rule_locally_automorphic_LS}
Consider a locally automorphic Lie system $(N,X,V)$. 
Then, $X$ admits a superposition rule that depends only on one particular solution of~$X$.
\end{corollary}

\begin{corollary}
\label{cor:superposition_rule_locally_automorphic_LS_2}
Let $(N,X,V^X)$ be a locally automorphic Lie system on a (connected) manifold. Then all $t$-independent constants of motion of $X$ are constants.
\end{corollary}

The existence of the local diffeomorphism $\varphi_x$ 
allows us to obtain theoretical properties 
of locally automorphic Lie systems,
such as Corollary \ref{cor:superposition_rule_locally_automorphic_LS}.
However,
its use to map them into automorphic Lie systems is quite limited
due to the locality of $\varphi_x$ and the difficulties in obtaining an explicit expression, which must be obtained by solving a system of nonlinear ordinary differential equations determined by the Lie algebra $V$.

This can be seen in the next two examples
of locally automorphic Lie systems. 

%%%%%%%%%%%%%%%%%%%%%%%%%%%%%%%%%%%%%%%%%%%%%%%%%%%%%%%%%%%%%%%%
\begin{example}
\textbf{(The generalised Darboux--Brioschi--Halphen (DBH) system \cite{Darboux})} Consider the system of ordinary differential equations given by
\begin{equation}
\label{Eq:Partial}
\begin{gathered}
\left\{
\begin{aligned}
\frac{{\rm d}w_1}{{\rm d}t}& = w_3w_2-w_1w_3-w_1w_2+\tau^2,\\
\frac{{\rm d}w_2}{{\rm d}t}& = w_1w_3-w_2w_1-w_2w_3+\tau^2,\\
\frac{{\rm d}w_3}{{\rm d}t}& = w_2w_1-w_3w_2-w_3w_1+\tau^2,\\
\end{aligned}\right.\\
\end{gathered}
\end{equation}
where
$$\tau^2=
\alpha_1^2(w_1-w_2)(w_3-w_1)+
\alpha_2^2(w_2-w_3)(w_1-w_2)+
\alpha_3^2(w_3-w_1)(w_2-w_3)\,,$$
and $\alpha_1,\alpha_2,\alpha_3\in \mathbb{R}$. The DBH system with $\tau=0$ appears in the description of triply orthogonal surfaces and the vacuum Einstein equations for hyper-K\"ahler Bianchi-IX metrics \cite{CH03,Darboux,Halphen}. 
Furthermore, the generalised DBH system for $\tau\neq 0$ is a
reduction of the self-dual Yang--Mills equations corresponding to an infinite-dimensional gauge group of diffeomorphisms of a
three-dimensional sphere \cite{CH03}. 

Even though the DBH system is autonomous, 
it is useful, 
e.g.\ to obtain its Lie symmetries \cite{EHLS-2014}, 
to view it as a Lie system related to a Vessiot--Guldberg Lie algebra $V^{\rm DBH} $ ($V$ to simplify the notation) spanned by the vector fields
\begin{align*}
X_1  &=
\frac{\partial}{\partial w_1} + \frac{\partial}{\partial w_2} + \frac{\partial}{\partial w_3}\,,\\
X_2  &=
w_1\frac{\partial}{\partial w_1} + w_2\frac{\partial}{\partial w_2} + w_3\frac{\partial}{\partial w_3}\,,
\\
X_3  &=
-(w_3w_2-w_1(w_3+w_2)+\tau^2)\frac{\partial}{\partial w_1} - (w_1w_3-w_2(w_1+w_3)+\tau^2)\frac{\partial}{\partial w_2}\\
&\quad - (w_2w_1-w_3(w_2+w_1)+\tau^2)\frac{\partial}{\partial w_3}\,.
\end{align*}
One can check that the commutation relations are, 
$$
[X_1 ,X_2 ] = X _1,
\qquad 
[X_1 ,X_3 ] = 2X_2 ,
\qquad 
[X_2 ,X_3 ] = X_3 .
$$
Thus, we have $V \simeq \mathfrak{sl}_2$, 
$\dim V =\dim \mathcal{O}$ and 
$X_1 \wedge X_2 \wedge X_3 \neq 0$ 
on an open submanifold $\mathcal{O}$ of $\mathbb{R}^3$. 
Hence, $(\mathcal{O},X,V )$ is a locally automorphic Lie system. 
To obtain a local diffeomorphism mapping the DBH system into an automorphic one, 
one needs to integrate the vector fields of $V $. In view of their analytic form, it is clear that it is very difficult to provide such a local diffeomorphism. \demo
\end{example}

%%%%%%%%%%%%%%%%%%%%%%%%%%%%%%%%%%%%%%%%%%%%%%%%%%%%%%%%%%%%%%%%
\begin{example}
\label{exampleCS}
(\textbf{A control system}
\cite{Ni00,Ra06}). 
Consider the following system of ordinary differential equations on~$\mathbb{R}^5$:
\begin{equation}
\label{Eq:ControlSys}
\begin{gathered}
\frac{{\rm d}x_1}{{\rm d}t}=b_1(t)\,,\quad 
\frac{{\rm d}x_2}{{\rm d}t}=b_2(t)\,,\quad
\frac{{\rm d}x_3}{{\rm d}t}=b_2(t)x_1\,,\quad 
\frac{{\rm d}x_4}{{\rm d}t}=b_2(t)x_1^2\,, \quad 
\frac{{\rm d}x_5}{{\rm d}t}=2b_2(t)x_1x_2\,,\
\end{gathered}
\end{equation}
where $b_1(t)$ and $b_2(t)$ are arbitrary $t$-dependent functions. 

This control system is defined by the $t$-dependent vector field 
$X^{\rm CS}=b_1(t)X_1+b_2(t)X_2$ on~$\mathbb{R}^5$,
where the vector fields
\begin{equation}
\label{Eq:BasisControl}
\begin{gathered}
X_1 = \frac{\partial}{\partial x_1}\,,
\qquad 
X_2 = \frac{\partial}{\partial x_2}+x_1\frac{\partial}{\partial x_3}+x_1^2\frac{\partial}{\partial x_4}+2x_1x_2\frac{\partial}{\partial x_5}\,,
\\[1ex]
X_3 = \frac{\partial}{\partial x_3}+2x_1\frac{\partial}{\partial x_4}+2x_2\frac{\partial}{\partial x_5}\,,
\qquad
X_4 = \frac{\partial}{\partial x_4}\,,
\qquad 
X_5 = \frac{\partial}{\partial x_5}\,,
\end{gathered}
\end{equation}
are such that their only non-vanishing commutation relations are
\begin{equation}
\label{Eq:ConRel}
[X_1,X_2]=X_3\,,\qquad [X_1,X_3]=2X_4\,,\qquad [X_2,X_3]=2X_5\,.
\end{equation}

It is remarkable that the initial system of differential equations (\ref{Eq:ControlSys}) can be described via a $t$-dependent vector field that can be written as a linear combination of the vector fields $X_1,X_2$ with $t$-dependent coefficients. Nevertheless,  to prove that $X$ is a Lie system, one has to add to $X_1,X_2$ as many vector fields as necessary to ensure that all such vector fields close a Vessiot--Guldberg Lie algebra. In fact,  the vector fields $X_3,X_4,X_5$ appear as successive Lie brackets of the vector fields $X_1,X_2$ and give rise, all of them, to a basis of a Vessiot--Guldberg Lie algebra.

Thus, $X_1,\ldots,X_5$ span a 5-dimensional nilpotent Lie algebra $V^{\rm CS}$. 
Since $X^{\rm CS}$ takes values in $V^{\rm CS}$, 
 we have that $(\mathbb{R}^5,X^{\rm CS},V^{\rm CS})$ is a Lie system \cite{Ra06}. 
The vector fields of $V^{\rm CS}$ span a distribution 
$\mathcal{D}^{V^{\rm CS}}={\rm T}\mathbb{R}^5$ 
and $\dim V^{\rm CS}=\dim \mathbb{R}^5$. 
Then, 
$(\mathbb{R}^5,X^{\rm CS},V^{\rm CS})$ 
becomes a locally automorphic Lie system. \demo
\end{example}

%%%%%%%%%%%%%%%%%%%%%%%%%%%%%%%%%%%%%%%%%%%%%%%%%%%%%%%%%%%%%%%%
\subsection{Invariants for locally automorphic Lie systems}\label{Sec:Inv}
%%%%%%%%%%%%%%%%%%%%%%%%%%%%%%%%%%%%%%%%%%%%%%%%%%%%%%%%%%%%%%%%

A {\it Lie symmetry of a Lie system $(N,X,V)$}
is a vector field $Y\in\mathfrak{X}(N)$ such that
$\mathcal{L}_{Y} Z =0$
for every vector field $Z \in V$.
If $V$ is spanned by the vector fields $X_1,\ldots,X_r$, this condition is equivalent to saying that $Y$ has to fulfil the following system of partial differential equations:
\begin{equation}\label{Eq:SysRos}
    \mathcal{L}_{X_i}Y = 0,\qquad i=1,\ldots,r.
\end{equation}

Note that the space of Lie symmetries of $(N,X,V)$ just depends on the Vessiot--Guldberg Lie algebra $V$, which motivates denoting it by ${\rm Sym}(V)$.  Moreover, ${\rm Sym}(V)$ is a Lie algebra.  In what follows, we study ${\rm Sym}(V)$ for the case of locally automorphic Lie systems  $(N,X,V)$. 

Consider a locally automorphic Lie system $(N,X,V)$. Then, each mapping $\varphi_x$ maps it to an automorphic Lie system $(G,X^R,V^R)$. It is clear that ${\rm Sym}(V^R)=V^L$. Since $\varphi_x$ is a local diffeomorphism mapping $V$ onto $V^R$, then it also maps ${\rm Sym}(V)$ onto $V^L$. Thus, we have the following lemma, whose implications are shown in Example \ref{Sym}.

\begin{lemma}\label{lemma:Exis}
   The Lie algebra of Lie symmetries of~$(N,X,V)$, namely ${\rm Sym}(V)$, is isomorphic to~$V$.
\end{lemma}

\begin{example}\label{Sym}
    Consider again Example \ref{exampleCS}, where we studied the control system given by (\ref{Eq:ControlSys}).

    Solving the linear system of partial differential equations (\ref{Eq:SysRos}) in the coefficients of a vector field $Y$ in the basis $\partial/\partial x_1,\dotsc,\partial/\partial x_5$, which demands a very long and tedious calculation, one sees that every Lie symmetry $Y$ of an arbitrary control system of the form (\ref{Eq:ControlSys}) has to be a linear combination with constant coefficients of the vector fields
\begin{equation}
\label{Eq:SymCS}
\begin{gathered}
Y_1 = \frac{\partial}{\partial x_1} + x_2\frac{\partial }{\partial x_3} + 2x_3\frac{\partial}{\partial x_4} + x_2^2\frac{\partial}{\partial x_5}\,, \qquad
Y_2 = \frac{\partial}{\partial x_2} + 2x_3\frac{\partial }{\partial x_5}\,,\\
Y_3 = \frac{\partial}{\partial x_3}\,, 
\qquad
Y_4 = \frac{\partial}{\partial x_4}\,,
\qquad
Y_5 = \frac{\partial}{\partial x_5}\,.
\end{gathered}
\end{equation}
It is easy to see that the vector fields $-Y_i$, where $i=1,\ldots,5$, 
span a Lie algebra with the same structure constants as $V^{\rm CS} = \langle X_1,\ldots,X_5\rangle$. \demo
\end{example}

From the fact that every locally automorphic Lie system $(N,X,V)$ is locally diffeomorphic to an automorphic Lie system $(G,X^R,V^R)$, 
we find that every differential form on $N$ that is invariant with respect to the Lie derivative of elements of~$V$ 
must be locally diffeomorphic to a left-invariant differential form on~$G$. 
Since ${\rm Sym}(V)$ is also diffeomorphic to $V^L$, 
and taking into account Remark \ref{Rem:Important},
one gets the following theorem: 
\begin{theorem}
\label{th:invariant_form} 
Consider a locally automorphic Lie system $(N,X,V)$. Let $Y_1,\ldots, Y_r$ be a basis of ${\rm Sym}(V)$, with a dual frame $\nu^1,\ldots, \nu^r$. 
Then, a differential form on~$N$ is invariant with respect to the 
%Vessiot--Guldberg 
Lie algebra $V$
if, and only if,
it is a linear combination with real coefficients of the exterior products of $\nu^1,\dotsc,\nu^r$.
\end{theorem}

\begin{remark} Let us consider a result on the invariants of general multisymplectic Lie systems. 
Let $(N,\Theta,X)$ be a multisymplectic Lie system $X$ with a Vessiot--Guldberg Lie algebra $V^X$.
A Lie algebra of vector fields $W$ on $N$ is a Lie algebra of symmetries of the Lie system $X$ if, and only if, $[V^X,W]=0$.
Let $Y_1,\ldots,Y_n$ be a basis of $W$. Then,
$$
\iota_{Y_1\wedge\ldots\wedge Y_n} \Theta
$$
are constants of the motion of our system.
\end{remark}

%%%%%%%%%%%%%%%%%%%%%%%%%%%%%%%%%%%%%%%%%%%%%%%%%%%%%%%%%%%%%%%%
\section{Reduction of multisymplectic Lie systems}
\label{section:reduction}
%%%%%%%%%%%%%%%%%%%%%%%%%%%%%%%%%%%%%%%%%%%%%%%%%%%%%%%%%%%%%%%%

This section introduces a reduction procedure for multisymplectic Lie systems by infinitesimal symmetries of the associated Lie system and its compatible multisymplectic form.
On the one hand, this simplifies the study of certain multisymplectic Lie systems by transforming them into new ones on manifolds of smaller dimension. On the other hand, the multisymplectic form of the reduced multisymplectic Lie system enables one to use the methods described in \cite{GLMV19} to study its properties. To introduce our multisymplectic Lie system reduction theory, we illustrate a reduction procedure based upon some examples.

%%%%%%%%%%%%%%%%%%%%%%%%%%%%%%%%%%%%%%%%%%%%%%%%%%%%%%%%%%%%%%%%
\subsection{Introductory examples}
%%%%%%%%%%%%%%%%%%%%%%%%%%%%%%%%%%%%%%%%%%%%%%%%%%%%%%%%%%%%%%%%

The following examples show how multisymplectic Lie systems can be reduced to new multisymplectic Lie systems on manifolds of smaller dimension. 

%%%%%%%%%%%%%%%%%%%%%%%%%%%%%%%%%%%%%%%%%%%%%%%%%%%%%%%%%%%%%%%%
\subsubsection{Schwarz equations}\label{subsubsection:SE}

Consider a Schwarz equation of the form \cite{Be07,OT09}
\begin{equation}\label{Eq:KS3}
\frac{\d^3x}{\d t^3}=\frac 32\left(
\frac{\d x}{\d t}\right)^{-1}\!\!\left(\frac{\d^2x}{\d t^2}\right)^{2}\!\!+2b_1(t)\frac{\d x}{\d t},
\end{equation}
where $b_1(t)$ is a non-constant function. The relevance of (\ref{Eq:KS3}) is due to its appearance in the study of Ermakov systems \cite{LA08} and the Schwarz derivative (see \cite{CGLS14} and references therein).

The differential equation (\ref{Eq:KS3}) is known to be a higher-order Lie system \cite{CGL11}. 
This means that the associated system of first-order differential equations obtained by adding the variables
$v:= \d x/\d t$ and $a:= \d^2x/\d t^2$, i.e.
\begin{equation}\label{Eq:firstKS3}
\frac{\d x}{\d t}=v\,,\qquad \frac{\d v}{\d t}=a\,,\qquad \frac{\d a}{\d t}=\frac 32 \frac{a^2}v+2b_1(t)v\,,
\end{equation}  
is a Lie system.
Indeed, (\ref{Eq:firstKS3}) is associated with the $t$-dependent vector field on $\mathcal{O}:=\{(x,v,a)\in{\rm T}^2\mathbb{R}\mid v\neq 0\}$ of the form
$
{X^{\rm S}}=X_3+b_1(t)X_1,
$ where the vector fields
\begin{equation}\label{Eq:VFKS1}
\begin{array}{c}
X_1=2v\dfrac{\partial}{\partial a}\,,\qquad X_2=v\dfrac{\partial}{\partial v}+2a\dfrac{\partial}{\partial a}\,,\qquad X_3=v\dfrac{\partial}{\partial x}+a\dfrac{\partial}{\partial v}+\dfrac 32
\dfrac{a^2}v\dfrac{\partial}{\partial a}\,,\end{array}
\end{equation}
satisfy
the commutation relations
\begin{equation}\label{Eq:KSbracket}
[X_1,X_2]=X_1\,,\quad [X_1,X_3]=2X_2\,, \quad [X_2,X_3]=X_3\,.
\end{equation}
Consequently, $X_1, X_2, X_3$ span a three-dimensional Lie algebra of vector fields $V^{\rm S}$ isomorphic
to $\mathfrak{sl}_2$ and $X^{\rm S}$ becomes a $t$-dependent vector field taking values in $V^{\rm S}$, i.e. $X^{\rm S}$ is a Lie system.

In \cite{GLMV19} it was proved that the vector fields of $V^{\rm S}$ are Hamiltonian relative to the multisymplectic structure $\Theta_{\rm S}$ on $\mathcal{O}$ given by
$\Theta_{\rm S}= \frac{1}{2v^3} \, \d a \wedge \d v \wedge \d x$  satisfying that 
\begin{equation}
\label{Eq:con}
\mathcal{L}_{X_\alpha}\Theta_{\rm S}=0\,,\qquad \alpha=1,2,3\,.
\end{equation}
Therefore,   $(\mathcal{O},\Theta_{\rm S},X^{\rm S})$ is a multisymplectic Lie system. Moreover,
\begin{equation}\label{Eq:HamSc}
\begin{gathered}
\iota_{X_1}\Theta_{\rm S}= \frac 1{v^2}\d v\wedge \d x =
\d\left(-\frac1v\d x\right)=\d\theta_1\,,
\\
\iota_{X_2}\Theta_{\rm S} = \frac 1{v}\left(\frac a{v^2}\d v-\frac 1{2v} \d a\right)\wedge \d x =
\d\left(-\frac{a}{2v^2}\d x+\frac{\d v}{2v}\right) =\d\theta_2\,,
\\
\iota_{X_3}\Theta_{\rm S} =
-\frac{3a^2}{4v^4}\d x\wedge \d v-\frac{a}{2v^3}\d a\wedge \d x+\frac 1{2v^2}\d a\wedge \d v =
\d\left( -\frac {a^2}{4v^3}\d x+ \frac{a}{v^2} \d v - \frac 1{2v}\d a \right)=
\d\theta_3\,.
\end{gathered}
\end{equation}
Hence, $X_1,X_2,X_3$ are Hamiltonian vector fields with respect to the multisymplectic manifold
% \footnote{Since $\Omega$ is a volume form is a multisymplectic structure}
$(\mathcal{O},\Theta_{\rm S})$,  with Hamiltonian one-forms
$$
\theta_1=-\frac{\d x}{v}\,,\qquad 
\theta_2=-\frac{a}{2v^2}\d x+\frac{\d v}{2v}\,,\qquad 
\theta_3=-\frac {a^2}{4v^3}\d x+ \frac{a}{v^2} \d v - \frac 1{2v}\d a\,.
$$

Consequently, no matter what the $t$-dependent coefficient $b_1(t)$ in (\ref{Eq:firstKS3}) is, the evolution of $X{\rm S}$ preserves the volume form $\Theta_{\rm S}$.
Since $\mathcal{D}^{V^{\rm S}}=T\mathcal{O}$, and in view of (\ref{Eq:con}), the value of $\Theta_{\rm S}$ at a point $o\in\mathcal{O}$ determines the value of $\Theta_{\rm S}$ on the connected component of $o$ in~$\mathcal{O}$.
Moreover, $\Theta_{\rm S}$ is,
up to a multiplicative constant on each connected component of $\mathcal{O}$,
the only volume form satisfying equations~(\ref{Eq:con}).
Since any one-form or two-form on a three-dimensional manifold is 1--degenerate, the system under study has a unique,
up to a non-zero proportional constant, multisymplectic form which is invariant under the action of $V^{\rm S}$.

The Schwarz equation, 
whose first-order system $X^{\rm S}$ is given by (\ref{Eq:firstKS3}), 
admits a Lie algebra of Lie symmetries, ${\rm Sym}(V^{\rm S})$, spanned by 
(see \cite{LG99,OT05})
\begin{equation}\label{Eq:SymSc2}
Y_1= \frac{\partial}{\partial x}
\,,\qquad
Y_2= x\frac{\partial}{\partial x} + v\frac{\partial}{\partial v} + a\frac{\partial}{\partial a}
\,,\qquad
Y_3= x^2\frac{\partial}{\partial x} + 2vx\frac{\partial}{\partial v} + 2(ax+v^2)\frac{\partial}{\partial a} \,.
\end{equation} 

Since $X^{\rm S}$ is a locally automorphic Lie system, ${\rm Sym}(V^{\rm S})$ is the Lie algebra of Lie symmetries given by Lemma \ref{lemma:Exis} and $[X_\alpha,Y_\beta]=0$, for $\alpha,\beta=1,2,3$. Therefore, $W:= \langle X_1, X_2, X_3, Y_1, Y_2, Y_3 \rangle$ is a Lie algebra isomorphic to $\mathfrak{sl}_2\oplus\mathfrak{sl}_2$.

A short computation shows that
$$
\iota_{Y_1}\Theta_S = \d \left( \frac1{4v^2}\d a \right) ,\qquad
\iota_{Y_2}\Theta_S = \d \left( -\frac{a}{2v} \d\left(\frac xv\right) \right) ,\qquad
\iota_{Y_3}\Theta_S =
\d \left( -\frac{x}{v}\d v - \frac{a}{2v} \d\left(\frac{x^2}{v}\right) \right) .
$$
Since the Vessiot--Guldberg Lie algebra for $X^{\rm S}$, namely (\ref{Eq:VFKS1}), consists also of Hamiltonian vector fields relative to $\Theta_{\rm S}$, the Lie algebra $W$ is made of Hamiltonian vector fields relative to $\Theta_{\rm S}$ and $\mathcal{L}_Z\Theta_{\rm S}=0$ for every $Z\in W$.

We denote by $G_2$ the one-parameter group of diffeomorphisms generated by the flow of the Lie symmetry $Y_2$ and the inversion $(x,v,a)\mapsto (-x-v,-a)$ on~$\mathcal{O}$. 
Define a projection $\pi:(x,v,a)\in \mathcal{O}\mapsto (\bar x = x/v,\bar a = a/v)\in \mathcal{O}/G_2$, where $\mathcal{O}/G_2$ is the space of orbits of the action of $G_2$ on~$\mathcal{O}$, which admits an immediate manifold structure as it is diffeomorphic to $\mathbb{R}^2$. Since $Y_2$ is a Lie symmetry of $X^{\rm S}$, the system $X^{\rm S}$ can be reduced onto $\mathcal{O}/G_2$. In coordinates, the projections of $X_1,X_2,X_3$ onto $\mathcal{O}/G_2$ read
$$
\bar{X_1} := 2\frac{\partial}{\partial \bar a}
\,,\quad
\bar{X_2} := -\bar x\frac{\partial}{\partial \bar x}+ \bar a\frac{\partial}{\partial \bar a}
\,,\quad
\bar{X_3} := (1-\bar x\bar a)\frac{\partial}{\partial \bar x}+\frac{\bar a^2}{2}\frac{\partial}{\partial \bar a} \,,
$$
respectively. They form a basis of the projection, $V^{\rm S}_R$, of the Lie algebra $V^{\rm S}=\langle X_1,X_2,X_3\rangle$ onto $\mathcal{O}/G_2$. In fact,
\begin{equation}\label{Eq:NewBasis}
[\bar X_1,\bar X_2]=\bar X_1,\qquad [\bar X_1,\bar X_3]=2\bar X_2,\qquad [\bar X_2,\bar X_3]=\bar X_3.
\end{equation}
Since $V^{\rm S}$ is simple and $V^{\rm S}_R\neq \{0\}$, one gets $V^{\rm S}_R\simeq V^{\rm S}\simeq\mathfrak{sl}_2$.

The vector fields $\bar{X_1},\bar{X_2},\bar{X_3}$ are Hamiltonian relative to the symplectic form defined by 
\begin{equation}\label{Eq:ReSym}
\bar{\Theta}:= \frac{1}{2}\d\bar x\wedge \d\bar a.
\end{equation}
In fact,
\[
\iota_{\bar{X_1}}\bar{\Theta} = -\d\bar{x}\,,\qquad \iota_{\bar{X_2}}\bar{\Theta} = -\frac{1}{2}\d(\bar{x}\bar{a})\,,\qquad \iota_{\bar{X_3}}\bar{\Theta} = \frac{1}{2}\d\Big(\bar{a} -\frac{\bar{x}\bar{a}^2}{2}\Big)\,.
\]

Moreover, $\bar{\Theta}$ is the only two-form in $\mathcal{O}/G_2\equiv\mathbb{R}^2$ such that $\pi^* \bar{\Theta}=\iota_{Y_2}\Theta_{\rm S}$. 

Thus, the multisymplectic Lie system 
$(\mathcal{O},\Theta_{\rm S},X^{\rm S})$ 
reduces to a new multisymplectic Lie system 
$(\mathcal{O}/G_2,\bar{\Theta},X^{\rm S}_R)$ 
with a Vessiot--Guldberg Lie algebra $V^{\rm S}_R$, where $X_R^{\rm S}$ is the projection of $X^{\rm S}$ onto $\mathcal{O}/G_2$. 
In particular, this reduced Lie system is a Lie--Hamilton system
(the multisymplectic form is a symplectic form).

To illustrate the simplification obtained by our multisymplectic reduction, let us comment on the particular reduced system $X^{\rm S}_R = \bar X_3 -1/4\bar X_1$, whose integral curves satisfy the system of differential equations
\begin{equation}\label{eq:Schwarz-reduced}
\frac{\d\bar x}{\d t} = 1 - \bar x\bar a\, ,\qquad \frac{\d \bar a}{\d t} = \frac{\bar a^2}{2} - \frac{1}{2}\,.
\end{equation}
Note that this system is a Hamiltonian system relative to the symplectic form (\ref{Eq:ReSym}). As a symplectic Hamiltonian system on the plane,  many of its properties, such as its superposition rules, have been studied elsewhere \cite{BBHLS13,BHLS15}. 
Let us comment on other properties that have not been studied so far, e.g.\ its equilibrium points. 
System (\ref{eq:Schwarz-reduced}) has equilibrium points at two points given by
$$\bar x=\bar a=\pm 1.
$$
It is notable that the system on $\mathcal{O}$ of the form
\begin{equation}\label{Eq:SchwarzParticular}
\frac{\d x}{\d t}=v\,,\qquad \frac{\d v}{\d t}=a\,,\qquad \frac{\d a}{\d t}=\frac 32 \frac{a^2}v-\frac{v}{2}\,,
\end{equation}
that projects onto (\ref{eq:Schwarz-reduced}) has no equilibrium point since $v\neq 0$ on~$\mathcal{O}$.  The points in $\mathcal{O}$ that project onto equilibrium points of (\ref{eq:Schwarz-reduced}) are called {\it relative equilibrium points} \cite{Am78,LZ21}.

Although analytical solutions for (\ref{Eq:SchwarzParticular}) can be obtained using mathematical manipulation software, the solutions obtained in this way are, as far as we were able to analyse them, too complicated to be easily studied. This justifies the execution of our reduction procedure. More precisely, the solutions to (\ref{eq:Schwarz-reduced}) are of the form
$$
\bar x(t)=(2c_2-1)(1+\cosh(t+2c_1))+\sinh(t+2c_1)\,,\qquad \bar a(t)= \frac{1-e^{2 c_1+t}}{e^{2 c_1+t}+1}\,, \qquad c_1,c_2\in \mathbb{R},
$$
for solutions with $|\bar a(t)|<1$ and 
$$
\bar x(t)=-1-2c_2+e^{t+2c_1}c_2+e^{-t-2c_1}(1+c_2)\,,\qquad \bar a(t)= \frac{1+e^{2 c_1+t}}{1-e^{2 c_1+t}},\qquad c_1,c_2\in \mathbb{R},
$$
for solutions for $|\bar a(t)|>1$. Meanwhile, the solutions for $\bar a(t)=\pm1$ read
\begin{equation}\label{Eq:ParSol}
\bar {x}(t)=\pm 1+e^{\mp t}c_2,
\end{equation}
respectively.

\begin{figure}[ht]
    \centering
    \includegraphics[width=0.3\textwidth]{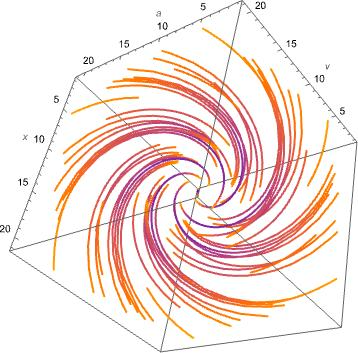}    \includegraphics[width=0.3\textwidth]{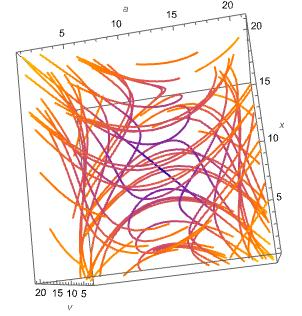}
    \includegraphics[width=0.3\textwidth]{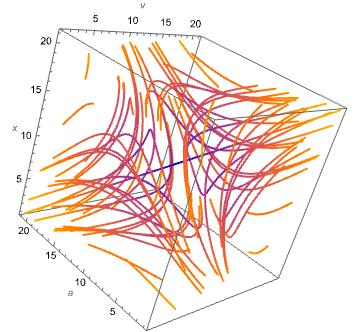}
    \includegraphics[width=0.4\textwidth]{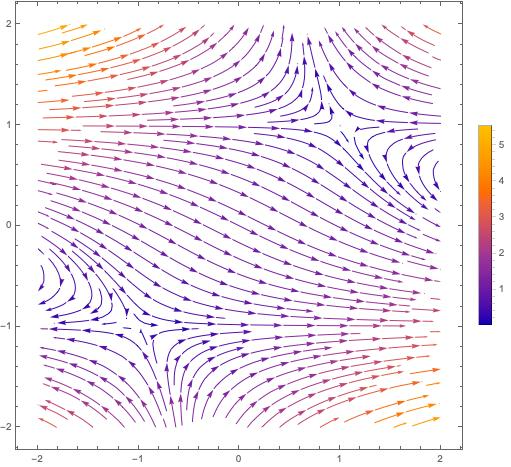}
    \caption{The first three diagrams represent some integral curves of the system (\ref{Eq:SchwarzParticular}) on the three-dimensional manifold $\mathcal{O}$ from different points of view, illustrating the fact that it has no stable points, as already commented. The figure on the second line depicts the integral curves of the reduced multisymplectic system (\ref{eq:Schwarz-reduced}) close to the equilibrium points $(1,1)$ and $(-1,-1)$.}\label{fig:Schwarz}
\end{figure}

It can be established whether two Lie systems in the plane related to a Vessiot--Guldberg Lie algebra isomorphic to $\mathfrak{sl}_2$, like (\ref{eq:Schwarz-reduced}), are locally diffeomorphic via the so-called {\it Casimir tensor fields} (see \cite{BHLS15} for details). In particular, it was proved in \cite[p.\,5]{BHLS15} that the Casimir tensor field for a Vessiot--Guldberg Lie algebra on the plane isomorphic to $\mathfrak{sl}_2$ in a basis with structure constants (\ref{Eq:NewBasis}) is given by
$$
G=X_1\otimes X_3+X_3\otimes X_1-2X_2\otimes X_2=-2\bar x^2\frac{\partial}{\partial \bar x}\otimes \frac{\partial}{\partial \bar x} + 2\frac{\partial}{\partial \bar x}\otimes \frac{\partial}{\partial \bar a}+2\frac{\partial}{\partial \bar a}\otimes \frac{\partial}{\partial \bar x}.
$$

It was also shown in \cite[Lemma 4.1]{BHLS15} that the determinant of the coefficients of $G$, namely $\det G=-4$, allows us to determine that there exists a local diffeomorphism on $\mathbb{R}^2$ mapping (\ref{eq:Schwarz-reduced}) onto a Lie system on the plane related to a Vessiot--Guldberg Lie algebra 
$$
\left\langle \frac{\partial}{\partial x}+\frac{\partial}{\partial y},x\frac{\partial}{\partial x}+y\frac{\partial}{\partial y},x^2\frac{\partial}{\partial x}+y^2\frac{\partial}{\partial y}\right\rangle,
$$
the so-called $\mathrm{I}_4$. It was proved in \cite[Table 1]{LL18}, that this Lie algebra is a Lie algebra of Killing vector fields relative to a pseudo-Riemannian metric on~$\mathbb{R}^2$. 
Physically, it is also notable that (\ref{eq:Schwarz-reduced}) is locally diffeomorphic to a Milne--Pinney equation \cite[Lemma 4.1 and Proposition 6.1]{BHLS15}, which appears in many physical problems (see \cite{CL09,Dissertationes,LS20} and references therein).

It is possible to see in Figure \ref{fig:Schwarz} that equilibrium points of (\ref{eq:Schwarz-reduced}) are unstable (solutions close to them move away from them as time passes by). 

There is a last remark about (\ref{eq:Schwarz-reduced}) to be considered. 
Since it is a Hamiltonian system, its evolution conserves the area of any set of solutions. 
Despite this, system (\ref{eq:Schwarz-reduced}) is still unstable around its equilibrium points.

%%%%%%%%%%%%%%%%%%%%%%%%%%%%%%%%%%%%%%%%%%%%%%%%%%%%%%%%%%%%%%%%
\subsubsection{Dissipative quantum harmonic oscillators}\label{subsubsection:DQHO}

Let us define $G^{do}:=SL_2\ltimes H_3$, where $SL_2$ is the Lie group of unimodular real $2\times 2$ matrices with real entries, $H_3$ is the Heisenberg group of upper real triangular unimodular $3\times 3$ matrices with ones in the diagonal, and $\ltimes$ represents the semidirect product of $SL_2$ acting on the normal subgroup $H_3$ of~$G^{do}$ via a Lie group morphism $\Phi:A\in SL_2  \mapsto \Phi_A\in{\rm Aut}(H_3)$, where ${\rm Aut}(H_3)$ is the space of Lie group automorphisms on $H_3$. Note that the Lie group product in $G^{do}$ can be written as
$$
(A,B)(A',B')=(AA',B\Phi_A(B')),\qquad A,A'\in SL_2,\qquad B,B'\in H_3.
$$
Note that the elements of $SL_2$ in $SL_2\ltimes H_3$ can be then written as $(A,I_3)$, where $A\in SL_2$ and $I_3$ is the neutral element of $H_3$. Meanwhile, the elements of $H_3$ in $SL_2\ltimes H_3$ can be brought into the form $(I_2,B)$ for some $B\in H_3$ and the neutral element $I_2$ in $SL_2$. Then,
$$
(A,I_3)(I_2,B)(A^{-1},I_3)=(A,\Phi_A(B))(A^{-1},I_3)=(I_2,\Phi_A(B)),\qquad A\in SL_2,\qquad B\in H_3.
$$
Hence, $\Phi_A(B)$ can be considered as the inner automorphism induced by $A$ acting on $B$ within $SL_2\ltimes H_3$. 
It is standard to consider $G^{do}$ as a matrix Lie group, at least locally around its neutral element. In this case, the elements of $SL_2$ and $H_3$ can be considered as subgroups of matrices in a certain $GL(n,\mathbb{R})$ and one can write $\Phi_A(B)=ABA^{-1}$ for every $A\in SL_2$ and $B\in H_3$.
The Lie algebra $\mathfrak{g}^{do}$ of~$G^{do}$ satisfies that $\mathfrak{g}^{do}\simeq\mathfrak{sl}_2\oplus_S\mathfrak{h}_3$,
where $\oplus_S$ represents a semi-direct sum of the Lie algebra
$\mathfrak{sl}_2\simeq T_eSL_2$
with the ideal
$\mathfrak{h}_3\simeq T_eH_3$ of $T_eG^{do}$.

Let us consider the automorphic Lie system on $G^{do}$ of the form
\begin{equation}\label{Eq:Partial1}
\frac{{\rm d}g}{{\rm d}t}=\sum_{\alpha=1}^6b_\alpha(t)X^R_\alpha(g)\,,\qquad g\in G^{do}\,,\qquad t\in \mathbb{R}\,,
\end{equation}
where $b_1(t),\ldots,b_6(t)$ are arbitrary $t$-dependent functions, and $\{X^R_1,\ldots,X^R_6\}$ form a basis of the Lie algebra $V^{do}$ of right-invariant vector fields on~$G^{do}$
so that $\{X^R_1,X_2^R,X_3^R\}$ is a basis of the Lie algebra $V_{sl}:=\langle X_1^R,X_2^R,X_3^R\rangle\simeq\mathfrak{sl}_2$
and
$\{X^R_4,X^R_5,X^R_6\}$ stands for a basis of the ideal
$V_h:=\langle X_4^R,X_5^R,X_6^R\rangle\simeq \mathfrak{h}_3$ of~$V^{do}$. Geometrically,  system (\ref{Eq:Partial1}) amounts to the $t$-dependent vector field on $G^{do}$ given by
\begin{equation}\label{Eq:DQHO}
X^{do}(t,g):=\sum_{\alpha=1}^6b_\alpha(t)X_\alpha^R(g)\,,\qquad g\in G^{do},\qquad t\in \mathbb{R}\,.
\end{equation}
Since the right-hand side of (\ref{Eq:DQHO}) is a $t$-dependent vector field $X^{do}$ on $G^{do}$ taking values in the finite-dimensional Lie algebra~$V^{do}$,
the system $X^{do}$ is a Lie system.

The relevance of system~(\ref{Eq:Partial1}) is due to the fact that it appears in the study of quantum harmonic oscillators with dissipation (see \cite[Eq. (3.9)]{Dissertationes} for details). More specifically, they appear in methods to solve  $t$-dependent Hamilton operators of the form
\begin{equation}\label{Eq:HamDis}
\widehat{H}(t)=\alpha(t)\frac{\widehat{x}^2}{2}+\beta(t)\frac{\widehat{p}^2}2+\gamma(t)\frac{\widehat{x}\widehat{p}+\widehat{p}\widehat{x}}{4}+\delta(t)\widehat{p}+\epsilon(t)\widehat{x}+\phi(t)\widehat{\mathrm{Id}},
\end{equation}
where $\alpha(t),\beta(t),\gamma(t),\delta(t),\epsilon(t),\phi(t)$ are arbitrary $t$-dependent functions, while $\widehat{x}, \widehat{p}, \widehat{\rm Id}$, are the position, moment, and identity operators on the space of complex square-integrable functions on $\mathbb{R}$, respectively. It is evident that (\ref{Eq:HamDis}) has many relevant applications in many quantum mechanical problems.

Although the relevant properties of $X^{do}$ are better understood using a geometric approach, a representation in coordinates of $X^{do}$ will also be useful to describe some of its features. 
More specifically, let us use the canonical coordinates of the second kind $\{v_1,\ldots,v_6\}$ on $G^{do}$ induced by (see \cite{Dissertationes} for details)
\begin{equation}
\label{Eq:SCPara}
\phi:\sum_{i=1}^6v_i{\rm a}^i\in \mathfrak{g}\mapsto \exp(-v_4{\rm a}^4)\exp(-v_5{\rm a}^5)\exp(-v_6{\rm a}^6)\exp(-v_1{\rm a}^1)\exp(-v_2{\rm a}^2)\exp(-v_3{\rm a}^3)\in G^{do},
\end{equation}
where $\{{\rm a}^1,\ldots,{\rm a}^6\}$ is a basis of $T_eG^{do}\simeq\mathfrak{g}^{do}$ spanning the commutation relations given in \cite[p.\,70]{Dissertationes} and $X_i^R(e)=-{\rm a}^i$ for $i=1,\ldots,6$.
This allows one to write system (\ref{Eq:Partial1}) in local coordinates as
\begin{equation}\label{Eq:WNQH}
  \begin{dcases}
    \frac{\d v_1}{\d t} = b_1(t)+b_2(t)\, v_1+b_3(t)\,v_1^2\,, & \frac{\d v_4}{\d t}=b_4(t)+\frac 12\, b_2(t)\, v_4+b_1(t)\,v_5\,,\\
    \frac{\d v_2}{\d t} = b_2(t)+2\,b_3(t)\,v_1\,, & \frac{\d v_5}{\d t}=b_5(t)-b_3(t)\, v_4-\frac 12\, b_2(t)\,v_5\,,\\
    \frac{\d v_3}{\d t} = e^{v_2}\,b_3(t)\,, & \frac{\d v_6}{\d t}=b_6(t)-b_5(t)\, v_4+\frac 12\, b_3(t)\,v_4^2-\frac 12\, b_1(t)\,v_5^2 \,.
  \end{dcases}
\end{equation}

Comparing the system $X^{do}$ in intrinsic form (\ref{Eq:Partial1}) with the above, we obtain
\begin{equation}\label{Eq:Rel}
\begin{aligned}
&X^R_1=\frac{\partial}{\partial v_1}+v_5\frac{\partial }{\partial v_4}-\frac{1}{2}v_5^2\frac{\partial }{\partial v_6}\,,
&X^R_2&=v_1\frac{\partial}{\partial v_1}+\frac{\partial }{\partial v_2}+\frac{1}{2}v_4\frac{\partial }{\partial v_4}-\frac{1}{2}v_5\frac{\partial }{\partial v_5}\,,\\
&X^R_3=v_1^2\frac{\partial }{\partial v_1}+2v_1\frac{\partial }{\partial v_2}+e^{v_2}\frac{\partial }{\partial v_3}-v_4\frac{\partial }{\partial v_5}+\frac{1}{2}v_4^2\frac{\partial }{\partial v_6}\,,
&X^R_4&=\frac{\partial }{\partial v_4}\,,
\\
&X^R_5=\frac{\partial }{\partial v_5}-v_4\frac{\partial }{\partial v_6}\,,
&X^R_6&=\frac{\partial }{\partial v_6}\,.
\end{aligned}
\end{equation}
The commutation relations between the above vector fields read
{\small
\begin{equation}\label{Eq:comG}
\begin{aligned}
&[X^R_1,X^R_2]=X^R_1\,, &&&&&&&&
\\
&[X^R_1,X^R_3]=2\, X^R_2\,,
&[X^R_2,X^R_3]&=X^R_3\,,
&& && &&
\\
&[X^R_1,X^R_4]=0\,,
&[X^R_2,X^R_4]&=-\frac 12\, X^R_4\,,
&[X^R_3,X^R_4]&=X^R_5\,,
&& &&
\\
&[X^R_1,X^R_5]=-X^R_4\,,
&[X^R_2,X^R_5]&=\frac 12\, X^R_5\,,
&[X^R_3,X^R_5]&=0\,,
&[X^R_4,X^R_5]&=-X^R_6\,,
&&
\\
&[X^R_1,X^R_6]=0\,,
&[X^R_2,X^R_6]&=0\,,
&[X^R_3,X^R_6]&=0\,,
&[X^R_4,X^R_6]&=0\,,
&[X^R_5,X^R_6]&=0\,.
\end{aligned}
\end{equation}
}

Therefore, the vector fields given in (\ref{Eq:Rel}) span a 6-dimensional real Lie algebra $V^{do}$ isomorphic to the semidirect sum of Lie algebras $\mathfrak{sl}_2\oplus_S\mathfrak{h}_3$.

Following the method described in \cite{GLMV19}, one can construct a multisymplectic form on $G^{do}$ turning the elements of $V^{do}$ into locally Hamiltonian vector fields. Since $(X^R_1\wedge\ldots\wedge X_6^R)(g)\neq 0$ on each $g\in G^{do}$, the vector fields $X^R_1,\ldots,X_6^R$ admit a family of differential one-forms $\eta^R_1,\ldots, \eta^R_6$ on $G^{do}$ such that $\eta^R_\alpha(X^R_\beta)=\delta_{\alpha\beta}$ for $\alpha,\beta=1,\ldots,6$, where $\delta_{\alpha\beta}$ is the Kronecker delta. More specifically,
the local expressions of these differential one-forms are
\[
\begin{array}{c}
\eta^R_1 = \d v_1 - v_1\d v_2 + v_1^2e^{-v_2}\d v_3\,,\,\,
\eta^R_2  = \d v_2 - 2v_1 e^{-v_2}\d v_3\,,\,\, \eta^R_3= e^{-v_2}\d v_3\,, \\\noalign{\medskip}
\eta^R_4 = -v_5\d v_1 + (v_1v_5-\frac{1}{2}v_4+\frac{1}{2}v_4v_5)\d v_2 + (v_1v_4-v_1^2v_5-v_1v_4v_5+v_4^2)e^{-v_2}\d v_3 + \d v_4 + v_4\d v_5\,,\\\noalign{\medskip}
\eta^R_5= \frac{1}{2}v_5\d v_2+(v_4-v_1v_5)e^{-v_2}\d v_3+\d v_5\,,\,\, \eta^R_6= \frac{1}{2}v_5\d v_1-\frac{1}{2}v_1v_5\d v_2+\frac{1}{2}(v_1^2v_5-v_4^2)e^{-v_2}\d v_3+\d v_6\,.
\end{array}
\]
As $\eta_1^R,\ldots,\eta_6^R$ are right-invariant differential forms, it makes sense to denote $V^{do*}:=\langle \eta^R_1,\ldots,\eta^R_6\rangle$.
This enables us to define the volume form
\begin{equation}\label{Eq:VolDis}
\Theta^{do}:=\eta^R_1\wedge\ldots\wedge \eta^R_6=e^{-v_2}\d v_1\wedge \ldots\wedge \d v_6\,.
\end{equation}
Since $\Theta^{do}$ is a volume form, it is clearly a multisymplectic one.

Let us prove that the vector fields $X^R_1,\ldots,X_6^R$ are locally Hamiltonian relative to~$\Theta^{do}$.
This is indeed a consequence of the geometric construction of $\Theta^{do}$ and the Lie algebra structure of~$\mathfrak{g}^{do}$.
We have that
$$
\left(\mathcal{L}_{X^R}\eta^R\right)(\bar{X}^R) =\mathcal{L}_{X^R}(\iota_{\bar X^R}\eta^R) -\eta^R([X^R,\bar{X}^R]) =-\eta^R([X^R,\bar{X}^R])\,,\quad
\forall X^R,\bar{X}^R\in V^{do}\,,\  \forall \eta^R\in V^{do*}\,.
$$
Since all the $\bar{X}^R$ span the tangent bundle to $G^{do}$, it follows that
$$
\mathcal{L}_{X^R}\eta^R=-{\rm ad}_{X^R}^*\eta^R,\qquad
\forall X^R\in V^{do}\,,\quad \forall \eta^R\in V^{do*}\,.
$$
The adjoint representation
${\rm ad} \colon X^R \in V^{do} \mapsto {\rm ad}_{X^R}:=[X^R,\cdot ] \in {\rm End}(V^{do})$
is such that ${\rm ad}(V^{do})$ consists of traceless operators on~$V^{do}$. Hence, the adjoint representation of $V^{do}$ is traceless and $V^{do}$ is a unimodular Lie algebra. Then, $V^{do}$ acts on $V^{do*}$ through the coadjoint Lie algebra representation
${\rm coad} \colon X^R \in V^{do*}\mapsto -{\rm ad}^*_{X^R} \in {\rm End}(V^{do*})$. Since ${\rm ad}$ is traceless, the representation ${\rm coad}$ is also traceless.
Thus,
$$
\mathcal{L}_{X^R}\Theta^{do}=-{\rm Tr}({\rm ad}_{X^R}^*)\Theta^{do}=0\,,\qquad \forall X^R\in V^{do}\,.
$$
As a consequence, the elements of $V^{do}$ are locally Hamiltonian vector fields with respect to the multisymplectic volume form~$\Theta^{do}$. The above construction gives rise to a Lie system on $G^{do}$ admitting a Vessiot--Guldberg Lie algebra of locally Hamiltonian vector fields $X^R_1,\ldots,X^R_6$ relative to the multisymplectic form~$\Theta^{do}$, that is $(G^{do}=SL_2\ltimes H_3,\Theta^{do},X^{do})$ is a multisymplectic Lie system.

The multisymplectic form, $\Theta^{do}$, given in (\ref{Eq:VolDis}), and the elements of $V^{do}$, detailed in (\ref{Eq:Rel}), are right-invariant. Hence,  for every left-invariant vector field $Y^L$ on~$G^{do}$, one has that $\mathcal{L}_{Y^L}\Theta^{do}=0$ and $\mathcal{L}_{Y^L}X^{do}=0$. Then, $Y_L$ is a Lie symmetry of $V^{do}$, that is,  $Y^L\in {\rm Sym}(V^{do})$. Lemma \ref{lemma:Exis} yields that $V^{ do}_L:={\rm Sym}(V^{do})$ is isomorphic to $V^{do}$. As shown next, this will lead to the construction of multisymplectic reductions of the multisymplectic Lie system $(G^{do},\Theta^{do},X^{do})$ by unimodular Lie subalgebras of $V^{  do}_L\simeq\mathfrak{sl}_2\oplus_S\mathfrak{h}_3$.

Let $\{X^L_1,\ldots,X^L_6\}$ be a basis of left-invariant vector fields on $G^{do}$ satisfying that
\begin{equation}\label{Eq:sameL}
X^L_\alpha(e)=X^R_\alpha(e)\,, \qquad \alpha=1,\ldots,6\,.
\end{equation}
A tedious calculation yields the local coordinate expression of these vector fields:
\begin{equation}\label{Eq:Rel2}
\begin{gathered}
X^L_1=e^{v_2}\frac{\partial}{\partial v_1}+2v_3\frac{\partial }{\partial v_2}+v_3^2\frac{\partial }{\partial v_3}\,,\quad
X^L_2=\frac{\partial}{\partial v_2}+v_3\frac{\partial }{\partial v_3}\,,\quad
X^L_3=\frac{\partial }{\partial v_3}\,,\\
X^L_4=e^{-v_2/2}(e^{v_2}-v_1v_3)\frac{\partial }{\partial v_4}-e^{-v_2/2}v_3\frac{\partial}{\partial v_5}-e^{-v_2/2}(e^{v_2}-v_1v_3)v_5\frac{\partial }{\partial v_6}\,,\\X^L_5=v_1e^{-v_2/2}\frac{\partial }{\partial v_4}+e^{-v_2/2}\frac{\partial}{\partial v_5}-v_1v_5e^{-v_2/2}\frac{\partial}{\partial v_6}\,,\quad X^L_6=\frac{\partial }{\partial v_6}\,.
\end{gathered}
\end{equation}
Note that $V^{do}_L=\langle X^L_1,\ldots,X^L_6\rangle$. The above properties will be useful in the following results.

%\paragraph{First multisymplectic reduction}
 As an example of reduction for the multisymplectic Lie system $(G^{do},\Theta^{do},X^{do})$, we consider the unimodular Lie subalgebra $V^L_{sl}:=\langle X^L_1,X^L_2,X^L_3\rangle\simeq\mathfrak{sl}_2$ and an associated Noether invariant $\Theta^{sl}_d:=\iota_{X^L_1\wedge X^L_2\wedge X^L_3}\Theta^{do}$
of the Lie system $X^{do}$.
The idea is to obtain a new multisymplectic Lie system $(\widetilde{G}, \widetilde{\Theta}^{sl}_d,\widetilde{X}^{sl}_d)$ such that the following conditions hold:
\begin{enumerate}
\item If $G^{do}/SL_2$ stands for the manifold of left cosets relative to the subgroup $SL_2$ of $G^{do}$,
then the canonical projection $\pi^{sl}\colon G^{do}\to \widetilde{G}:=G^{do}/SL_2$ is a submersion and $\widetilde{G}$ becomes a manifold, which does not need to have a Lie group structure.
\item
The vector fields of the Vessiot--Guldberg Lie algebra $V^{do}$ project onto a Lie algebra $\widetilde{V}^{sl}_d:=(\pi^{sl})_*(V^{do})$ of vector fields on~$\widetilde{G}$.
\item  Finally,
$(\pi^{sl})^*\widetilde{\Theta}^{sl}_d={\Theta}^{sl}_d=\iota_{X^L_1\wedge X^L_2\wedge X^L_3}\Theta^{do}$.
\end{enumerate}

The action of $SL_2$ on $G^{do}$ on the right has fundamental vector fields $\langle X_1^L,X_2^L,X_3^L\rangle$. The orbits of this action are the left cosets of $SL_2$ in $G^{do}$. Furthermore, the quotient space $G^{do}/SL_2$ is diffeomorphic to $H_3$. 
Thus, the first condition is satisfied.

Let us verify the second condition.
Since the elements of $V^{do}$ commute with left-invariant vector fields on~$G^{do}$, they are all $\pi^{sl}$-projectable vector fields onto $G^{do}/SL_2$.
As $V^{do}=\langle X^R_1,\ldots,X_6^R\rangle$ is a Lie algebra whose distribution, $\mathcal{D}^{V^{do}}$, is such that $\mathcal{D}^{V^{do}}_g=T_gG^{do}$ for every $g\in G^{do}$, the projections of $X^R_1,\ldots,X_6^R$ onto $G^{do}/SL_2\simeq H_3$ span $T(G^{do}/SL_2)$.
Hence, the initial system $X^{do}$ can be projected onto $G^{do}/SL_2$ giving rise to a Lie system on $G^{do}/SL_2$ with a Vessiot--Guldberg Lie algebra  $\widetilde{V}_d^{sl}=(\pi^{sl})_{*}(V^{do})$.
The morphism $(\pi^{sl})_*\vert_{V^{do}}:V^{do}\rightarrow \widetilde{V}^{sl}_d$ is a surjective Lie algebra morphism.
%Hence, $V^{do}$ is a Lie algebra extension of $\widetilde{V}^{sl}$.

Note that $v_1,\ldots,v_6$ form a canonical coordinate system of the second kind defined locally in $G^{do}$ whilst, in view of our parametrisation (\ref{Eq:SCPara}), the  $v_4,v_5,v_6$ are, indeed global, coordinates in  $G^{do}/SL_2$. 
In coordinates $v_1,\ldots,v_6$, we have the canonical projection
$\pi^{sl} \colon (v_1,\ldots,v_6) \in G^{do} \mapsto (v_4,v_5,v_6) \in G^{do}/SL_2\simeq H_3$.
The projections of the elements of the basis (\ref{Eq:Rel}) of the Vessiot--Guldberg Lie algebra $V^{do}$ of $X^{do}$ onto $G^{do}/SL_2$ read
\begin{equation}\label{Eq:Rel3}
\begin{aligned}
\widetilde{X}^{sl}_1&=v_5\frac{\partial }{\partial v_4}-\frac{1}{2}v_5^2\frac{\partial }{\partial v_6}\,,
&\widetilde{X}^{sl}_2&=\frac{1}{2}v_4\frac{\partial }{\partial v_4}-\frac{1}{2}v_5\frac{\partial }{\partial v_5}\,,
&\widetilde{X}^{sl}_3&=-v_4\frac{\partial }{\partial v_5}+\frac{1}{2}v_4^2\frac{\partial }{\partial v_6}\,,\\
\widetilde{X}^{sl}_4&=\frac{\partial }{\partial v_4}\,,
&\widetilde{X}^{sl}_5&=\frac{\partial }{\partial v_5}-v_4\frac{\partial }{\partial v_6}\,,
&\widetilde{X}^{sl}_6&=\frac{\partial }{\partial v_6}\,.
\end{aligned}
\end{equation}

A simple calculation shows that
$$
[\widetilde{X}^{sl}_\alpha,\widetilde{X}^{sl}_\beta]=c_{\alpha\beta}^\gamma\widetilde{X}^{sl}_\gamma\,, \qquad \alpha,\beta,\gamma=1,\ldots, 6\,,
$$
where the $c_{\alpha\beta}^\gamma$ are again the same constants given in (\ref{Eq:comG}) for $V^{do}$.
This shows that $\widetilde{V}^{sl}_{d}\simeq V^{do}\simeq \mathfrak{g}^{do}$.
This gives rise to a Lie system $X^{G^{do}}_{sl}:=\sum_{\alpha=1}^6b_\alpha(t)\widetilde{X}^{sl}_\alpha$ on the quotient space $G^{do}/SL_2$.

Finally, let us prove that there exists a multisymplectic form $\widetilde{\Theta}^{sl}_d$ on $G^{do}/SL_2$ turning the elements of $\widetilde{V}^{sl}_d$ into locally Hamiltonian vector fields. Since $V_{sl}^L$ is a Lie subalgebra, it makes sense to define the restriction
${\rm ad}_{X^L_\alpha}|_{V_{sl}^L}$
of each
${\rm ad}_{X^L_\alpha} \colon V_L^{do} \rightarrow V_L^{do}$,
with $\alpha=1,\ldots,6$, to $V_{sl}^L$.
In view of the commutation relations (\ref{Eq:comG}) of the vector fields of $V^{do}$, it follows that ${\rm Tr}({\rm ad}_{X^L_\alpha}|_{V_{sl}^L})=0$. Therefore,
\begin{equation}
\label{Eq:1}
\mathcal{L}_{X_\alpha^L}\big(\iota_{X^L_1\wedge X^L_2\wedge X^L_3}\Theta^{do}\big)={\rm Tr} \big({\rm ad}_{X_\alpha^L}|_{V^L_{sl}}\big)\iota_{X_1^L\wedge X_2^L\wedge X_3^L}\Theta^{do}=0\,,\qquad \alpha=1,2,3\,.
\end{equation}
Moreover,
$
\Theta^{sl}_d:=\iota_{X^L_1\wedge X^L_2\wedge X^L_3}\Theta^{do}
$
satisfies
\begin{equation}
\label{Eq:2}
\big(\ker \Theta^{sl}_d\big)(g)=\langle X^L_1(g),X^L_2(g),X^L_3(g)\rangle\,,\qquad \forall g\in G^{do}\,.
\end{equation}
Expressions (\ref{Eq:1}) and (\ref{Eq:2}) ensure that there exists a unique well-defined differential form $\widetilde{\Theta}^{sl}_d$ on $G^{do}/SL_2$ such that $(\pi^{sl})^*\widetilde{\Theta}^{sl}_d=\Theta^{sl}_d$ relative to the canonical projection $\pi^{sl}:G^{do}\rightarrow G^{do}/SL_2$. As the vector fields of $\widetilde{V}^{sl}_d$ span the tangent bundle to $G^{do}/SL_2$ and $\widetilde \Theta^{sl}_d$ is invariant under the action by Lie derivatives of the elements of $\widetilde{V}^{sl}_d$ on it, then $\widetilde \Theta^{sl}_d$ is the only multisymplectic form, up to a non-zero proportionality constant, turning the vector fields of $\widetilde{V}^{sl}_d$ into locally Hamiltonian vector fields relative to $\widetilde \Theta^{sl}_d$ .
Since $\Theta^{do}$ is a volume form on a 6-dimensional manifold, then $\widetilde{\Theta}^{sl}_d$ is a differential three-form. Therefore, $\widetilde \Theta^{sl}_d$ is a volume form on $G^{do}/SL_2$ and 
$
\d\widetilde \Theta^{sl}_d=0.
$
In view of (\ref{Eq:2}), it follows that $\widetilde\Theta^{sl}_d$ is nondegenerate. Hence,  $(G^{do}/SL_2,\widetilde \Theta^{sl}_d)$ is a multisymplectic manifold.

The differential form $\widetilde \Theta^{sl}_d$ can be explicitly derived by taking into account that it is, up to a non-zero constant, invariant by the above projected vector fields (\ref{Eq:Rel3}) spanning the Lie algebra $\widetilde{V}_d^{sl}$, namely  $\widetilde\Theta^{sl}_d=\lambda \d v_4\wedge \d v_5\wedge \d v_6 $  for a certain $\lambda \in \mathbb{R}\backslash \{0\}$ and $\iota_{X^L_4\wedge X^L_5\wedge X_6^L}(\pi^{sl})^*\widetilde{\Theta}^{sl}_d=1$ when $v_4=v_5=v_6=0$. Hence,
$$
\widetilde\Theta^{sl}_d=\d v_4\wedge \d v_5\wedge \d v_6 \,.
$$
It is worth noting that the multisymplectic reduction process has been carried out  assuming $V^L_{sl}$ to be unimodular.

%%%%%%%%%%%%%%%%%%%%%%%%%%%%%%%%%%%%%%%%%%%%%%%%%%%%%%%%%%%%%%%%
\subsection{Multisymplectic reduction of multisymplectic Lie systems}
%%%%%%%%%%%%%%%%%%%%%%%%%%%%%%%%%%%%%%%%%%%%%%%%%%%%%%%%%%%%%%%%

After studying the above examples, the following question arises: Can the above process be repeated for other multisymplectic Lie subalgebras of $V^{do}$? To answer this question, we define the momentum map
$$
\begin{array}{rcccl}
\mathcal{J}:&G\times \mathcal{G}(\mathfrak{g})&\longrightarrow &\bigwedge^\bullet T^*G\\
&(g,X_1^L\wedge\ldots\wedge X_p^L)&\longmapsto&(\iota_{X^L_1\wedge\ldots\wedge X^L_p}\Theta)(g)\,, & %\in\bigwedge^\bullet T^*_gG\,,
\end{array}
$$
where $\mathcal{J}(g,X_1^L\wedge\ldots X_p^L)\in\bigwedge^\bullet T^*_gG$, 
the set $\mathcal{G}(\mathfrak{g})$  is the unimodular Grassmannian defined in Definition \ref{def:grassman} and $\Theta$ is a multisymplectic form on~$G$.

In the example in Section \ref{subsubsection:SE}, we considered the space $\mathcal{J}(g,Y_2)$, while $\mathcal{J}(g,X^L_1\wedge X^L_2\wedge X^L_3)$ and $\mathcal{J}(g,X^L_4\wedge X^L_5\wedge X^L_6)$ were used in the two multisymplectic reductions of the example in Section \ref{subsubsection:DQHO}. In both cases, $\mathfrak{g}$ is the abstract Lie algebra isomorphic to the Lie algebra of Lie symmetries ${\rm Sym}(V^X)$, where $V^X$ is the Vessiot--Guldberg Lie algebra of the multisymplectic Lie system. Furthermore, it can be shown that in both cases
$$
\mathcal{L}_{X}\mathcal{J} = 0\,,
\qquad \forall X\in V^X\,.
$$

This section aims to use the examples of reduction of multisymplectic Lie systems already studied to provide a general multisymplectic reduction procedure for multisymplectic Lie systems.

%%%%%%%%%%%%%%%%%%%%%%%%%%%%%%%%%%%%%%%%%%%%%%%%%%%%%%%%%%%%%%%%
\subsubsection{Multisymplectic reduction}

Let $(N, \Theta, X)$ be a multisymplectic Lie system, where $(N, \Theta)$ is a multisymplectic manifold and $X$ denotes a Lie system on~$N$. Let $V^X$ be a Vessiot--Guldberg Lie algebra of $X$ and denote by ${\rm Sym}(V^X)$ the finite-dimensional Lie algebra of Lie symmetries of a Lie system $(N,X,V^X)$. Consider $\mathfrak{g}$ to be the abstract Lie algebra isomorphic to ${\rm Sym}(V^X)$.

To obtain a general reduction procedure for multisymplectic Lie systems, one considers the map
$$
\begin{array}{rcccl}
\mathcal{J}:&N\times \mathcal{G}(\mathfrak{g})&\longrightarrow &\bigwedge^\bullet T^*N\\
&(x,w)&\longmapsto&(\iota_{w}\Theta)(x) \,,
\end{array}
$$
where $(\iota_{w}\Theta)(x)\in \bigwedge^\bullet T^*_xN\,$ for every $x\in N$ and $w\in\mathcal{G}(\mathfrak{g})$.

Let $V$ be an $r$-dimensional unimodular Lie subalgebra of ${\rm Sym}(V^X)$ and let $\mathcal{D}^V$ be the distribution generated by the vector fields of $V$.  We consider the reduction of the original system by~$V$.
The reduced system is a multisymplectic Lie system if:
\begin{enumerate}
\item
The quotient space $N/\mathcal{D}^V$ is a manifold and the canonical projection $\pi\colon N\to N/\mathcal{D}^V$ is a submersion.
\item
The elements of a basis of $V$ are $\pi$-projectable.
\item
There exists a multisymplectic form $\widetilde{\Theta}$ on $N/\mathcal{D}^V$
such that
$(\pi^*\widetilde{\Theta})(x)=\mathcal{J}(x,X_1\wedge\dotsb \wedge X_r)$, for every $x\in N$ and where $X_1\wedge\dotsb \wedge X_r$ is the unique (up to an irrelevant multiplicative non-zero constant) $r$-vector field generated by the elements of $V$.
\end{enumerate}

The idea is to give conditions in terms of $\Theta$, $V$, and $\mathcal{J}$ so that these three properties are fulfilled. The examples in the previous section can be considered as an application of the following reduction procedure.

\begin{theorem}\label{th:RedMul} Let $\Theta$ be an $(r+1)$-nondegenerate multisymplectic form on $N$ and let $G$ be an $r$-dimensional unimodular Lie group. We consider a proper and free unimodular Lie group action 
$\Phi \colon G \times N \rightarrow N$ and denote by $V_\Phi$ the Lie algebra of fundamental vector fields of $\Phi$. If 
 $w\in \Lambda^{r}V_\Phi\backslash\{0\}$ 
is a Hamiltonian $r$-vector field with respect to the multisymplectic form $\Theta$, then
there exists a unique multisymplectic form $\widetilde \Theta$ on $N/G$ given by
$$
\pi^*\widetilde{\Theta} = \iota_w\Theta\,
$$
where $\pi\colon N \to N/G$ is the quotient projection of $G$ onto  $N/G$.
\end{theorem}
\begin{proof}
The differential form $\iota_w\Theta$ can be projected onto $N/G$ if its contraction with the elements of $V_\Phi$ vanishes and $\iota_w\Theta$ is invariant relative to the elements of $V_\Phi$. The first condition is ensured by the fact that $w\in \Lambda^{r}V_\Phi$ and therefore
$
\iota_{X}\iota_w\Theta=\iota_{w\wedge X}\Theta=0$ for all $X\in V_\Phi$ because $X\wedge w=0$. 
 The second condition follows from the following considerations. By assumption, the vector fields $X_1,\ldots,X_r$ are $\Theta$-Hamiltonian, i.e. $\mathcal{L}_{X}\Theta=0$ for every $X\in V_\Phi$, and $V_\Phi$ is a unimodular Lie algebra of vector fields, namely $\mathcal{L}_Xw= 0$ for all $X\in V_\Phi$. This implies that
$$
\mathcal{L}_{X}\iota_w\Theta =
(-1)^{r}\iota_w\mathcal{L}_{X}\Theta+\iota_{\mathcal{L}_Xw}\Theta=0\,,
\qquad \forall X\in V_\Phi\,.
$$
In consequence, $\mathcal{L}_X\Theta=0$ and $\iota_X\iota_w\Theta=0$  for all $X\in V_\Phi$, which implies that there exists a differential form $\widetilde{\Theta}$ on $N/G$ such that $\pi^*\widetilde{\Theta}=\iota_w\Theta$.
The differential form $\widetilde \Theta$ is also unique because
$\pi^* \colon \Omega^\bullet(N/G)\rightarrow \Omega^\bullet(N)$
is an injection.

Moreover, $\widetilde{\Theta}$ is closed since the assumption that $w$ is Hamiltonian enables us to write that
$$
0=\mathcal{L}_w\Theta = (\d\iota_w+(-1)^{r}\iota_w\d)\Theta = \d\iota_w\Theta = \d\pi^*\widetilde{\Theta} = \pi^*\d\widetilde \Theta
\quad\Longrightarrow\quad
\d\widetilde \Theta=0\,.
$$

Since $\Theta$ is $(r+1)$-nondegenerate and $\Phi$ is effective, then $w\neq 0$ and the form $\iota_w\Theta$ is one-nondegenerate. Hence, $\widetilde{\Theta}$ is also one-nondegenerate. This turns $(N/G,\widetilde{\Theta})$ into a multisymplectic manifold.
\end{proof}

\begin{definition} A pair $(\Theta,\Phi)$ satisfying the conditions of the Theorem \ref{th:RedMul} is called a {\it multisymplectic reduction scheme}.
  \end{definition}

Once a reduction scheme is provided, it is mandatory to show how this can be applied to the reduction of a multisymplectic Lie system.
\begin{theorem} \label{th:RedLie}Let $(N,\Theta,X)$ be a multisymplectic Lie system and let $(\Theta,\Phi)$ be a multisymplectic reduction scheme whose $\{\Phi_g\}_{g\in G}$ are symmetries of  the elements of $V^X$. Then, $X$ can be projected onto $N/G$ giving rise to a new multisymplectic Lie system $(N/G,\widetilde \Theta, \pi_*X)$, where $\widetilde{\Theta}$ is the multisymplectic form on $N/G$ induced by the multisymplectic reduction scheme $(\Theta,\Phi)$.
  \end{theorem}
\begin{proof} If the $\Phi_g$, with $g\in G$, are symmetries of the elements of $V^X$, in particular they are symmetries of $X$ because it takes values in~$V^X$. Hence, $X$ can be projected onto $N/G$. If $\pi:N\rightarrow N/G$ is the quotient projection, then $\pi_*|_{V^X}:V^X\rightarrow \mathfrak{X}(N/G)$ is a Lie algebra morphism and $\pi_*(V^X)$ is therefore finite-dimensional turning $\pi_*X$ into a Lie system. Since $(N,\Theta,X)$ is a multisymplectic Lie system and the elements of $V^X$ commute with the fundamental vector fields of the Lie group action $\Phi$, one has that $\mathcal{L}_X\iota_{X_1\wedge\ldots \wedge X_r}\Theta=0$ and, as $X$ and $\Theta$ are projectable onto $N/G$, then $\mathcal{L}_{\pi_*X}\widetilde \Theta=0$. Therefore $(N/G,\widetilde \Theta,\pi_*X)$ is a multisymplectic Lie system.
  \end{proof}

It is frequently accessible to find Lie symmetries for a Lie system by using Lie theory methods. Nevertheless, the crux is to find a multisymplectic form satisfying the conditions provided in Theorem \ref{th:RedLie}.

%%%%%%%%%%%%%%%%%%%%%%%%%%%%%%%%%%%%%%%%%%%%%%%%%%%%%%%%%%%%%%%
\begin{subsubsection}{Reductions of a control system}
%%%%%%%%%%%%%%%%%%%%%%%%%%%%%%%%%%%%%%%%%%%%%%%%%%%%%%%%%%%%%%%

Consider again the control system given in Example \ref{exampleCS}. Let us now apply our multisymplectic reduction to this control system.
Recall that the vector fields $X_1,\ldots,X_5$ span, in view of their unique non-vanishing commutation relations (\ref{Eq:ConRel}), a unimodular Lie algebra $V$. Moreover, $\mathcal{D}^V={\rm T}\mathbb{R}^5$. In consequence, this control system is a locally automorphic one and the application to it of the techniques devised in \cite{GLMV19} gives rise to a multisymplectic Lie system $(\mathbb{R}^5,\Theta^B,X)$. For example, we can construct the multisymplectic structure
$$
\Theta^B=\eta_1\wedge\eta_2\wedge\eta_3\wedge\eta_4\wedge\eta_5\,,
$$
where $\eta_1,\ldots,\eta_5$ is a dual basis to $X_1,\ldots,X_5$ given by
% In coordinates,
$$
 \begin{gathered}
 \eta_1 = \d x_1\,,\qquad
 \eta_2 = \d x_2\,,\qquad
 \eta_3 = -x_1\d x_2+\d x_3 \,,
 \\
 \eta_4 = x_1^2\d x_2-2x_1\d x_3+\d x_4\,,\qquad
 \eta_5 = -{2x_2}\d x_3+\d x_5 \,.
 \end{gathered}
 $$
In coordinates,
$$
\Theta^B = \d x_1\wedge\dotsb\wedge \d x_5 \,.
$$

Let us use the Lie symmetries of this system to accomplish several of its possible reductions. Recall that the Lie symmetries of the control systems (\ref{Eq:ControlSys}) are given in (\ref{Eq:SymCS}).
 Indeed, after a long but simple calculation, we obtain a Lie algebra $V^S$ of Lie symmetries spanned by
 $$
 \begin{gathered}
 Y_1=\frac{\partial}{\partial x_1}+x_2\frac{\partial}{\partial x_3}+2x_3\frac{\partial}{\partial x_4}+x_2^2\frac{\partial}{\partial x_5},\qquad Y_2=\frac{\partial}{\partial x_2}+2x_3\frac{\partial}{\partial x_5},\qquad
 \\
 Y_3=\frac{\partial}{\partial x_3},\qquad Y_4=\frac{\partial}{\partial x_4},\qquad Y_5=\frac{\partial}{\partial x_5}.
 \end{gathered}
 $$
Recall also that
$$
Y_\alpha(0)=X_\alpha(0),\qquad \alpha=1,\ldots,5
$$
and the only non-vanishing commutation relations between the previous vector fields read
\[
[Y_1,Y_2]=-Y_3,\qquad [Y_1,Y_3]=-2Y_4,\qquad [Y_2,Y_3]=-2Y_5.
\]
Hence,  $V^S=\langle Y_1,\ldots,Y_6\rangle\simeq V$.
Let $\mathfrak{g}_C$ be the abstract Lie algebra isomorphic to $V$.
This leads us to define a momentum map
$$
\begin{array}{rcccc}
\mathcal{J} \colon & \mathbb{R}^5\times \mathcal{G}(\mathfrak{g}_C) &
\longrightarrow & \bigwedge^\bullet T^*\mathbb{R}^5
\\
& (\xi,w) &
\longmapsto & \iota_{w}\Theta^B(\xi) \,.
\end{array}
$$

Consider the unimodular subalgebra $V^L_{45}:=\langle Y_4,Y_5\rangle$ and the associated Noether invariant of $X^G$:
$$
\bar\Theta_{45} := \iota_{Y_4\wedge Y_5}\Theta^B={ \d x_1}\wedge \d x_2\wedge \d x_3 \,.
$$

The Lie algebra $\langle Y_4,Y_5\rangle \simeq \mathbb{R}^2$ can be considered as the Lie algebra of fundamental vector fields of a proper and free unimodular Lie group action
$$
\Phi_{45}:(\lambda_1,\lambda_2;x_1,x_2,x_3,x_4,x_5)\in \mathbb{R}^2\times \mathbb{R}^5\longmapsto (x_1,x_2,x_3,x_4+\lambda_1,x_5+\lambda_2)\in \mathbb{R}^5.
$$
Moreover, the space $\Lambda^2V_\Phi=\langle Y_4\wedge Y_5\rangle$ is Hamiltonian. Since $\Theta^B$ is 5-nondegenerate, the pair 
$(\Theta^B,\Phi_{45})$ 
becomes a reduction scheme and Theorem \ref{th:RedMul} can be applied to reduce this multisymplectic structure to $\mathbb{R}^5/\mathbb{R}^2$. 
More specifically, the orbits of the action $\Phi$ take the form 
$\mathcal{O}_{[x_1,x_2,x_3,x_4,x_5]}=\{x_1\}\times\{x_2\}\times\{x_3\}\times \mathbb{R}^2$ 
and the quotient space is $\mathbb{R}^5/\mathbb{R}^2 \simeq \mathbb{R}^3$. 
The variables $x_1,x_2,x_3$ give rise to a coordinate system in the quotient and  one can define the submersion
$$
\pi_{45} \colon (x_1,\ldots,x_5) \in \mathbb{R}^5 \longmapsto (x_1,x_2,x_3) \in \mathbb{R}^3.
$$
Additionally, Theorem \ref{th:RedMul} ensures that there exists a unique multisymplectic form $\widehat \Theta_{45}$ on~$\mathbb{R}^3$ whose pull-back to $\mathbb{R}^5$ is $\bar \Theta_{45}$, namely $\widehat \Theta_{45} = \d x_1 \wedge \d x_2 \wedge \d x_3$.

Since $(\Theta^B,\Phi_{45})$ is a reduction scheme such that the $\Phi_g$, with $g\in
\mathbb{R}^2$, are symmetries of the Lie system of the multisymplectic Lie system $(\mathbb{R}^5,\Theta^B,X)$, Theorem \ref{th:RedLie} ensures that we can reduce the multisymplectic Lie system. 
In fact, the Vessiot--Guldberg Lie algebra for the control system (\ref{Eq:ControlSys}) project onto $\mathbb{R}^3$  giving rise to  a new Lie algebra $V_{45}$ spanned by vector fields
$$
X_1^{45}=\frac{\partial}{\partial x_1}\,,\qquad X_2^{45}=\frac{\partial}{\partial x_2}+x_1\frac{\partial}{\partial x_3}\,,\qquad X_3^{45}=\frac{\partial}{\partial x_3}\,.
$$
Hence, also the system $X^{CS}$ can be projected onto $\mathbb{R}^3$ to give rise to a Lie system $X_{45}$ on~$\mathbb{R}^3$. 
Moreover, the Lie algebra $V_{45}$ consists of Hamiltonian vector fields relative to $\widehat \Theta_{45}$, which turns the triple $(\mathbb{R}^3,\widehat \Theta_{45},X_{45})$ into a new multisymplectic Lie system.

\end{subsubsection}
\begin{subsubsection}{Systems on Lie groups and quantum  harmonic oscillators  with a spin-magnetic term}\label{SubSec:OsSM}

  Let us consider a Lie system on $G^{os}:=SL(2,\mathbb{R})\times SO(3)$ related to a quantum mechanical system given by a $t$-dependent harmonic oscillator with a spin-magnetic term. This example will be useful in the next section to illustrate how several reductions of a multisymplectic Lie system can be used to reconstruct the dynamics of the original one. 

  The Lie algebra of $G^{os}$ takes the form $\mathfrak{g}^{os}:=\mathfrak{sl}(2,\mathbb{R})\oplus\mathfrak{so}(3)$ where $\oplus$ represents the direct sum of the Lie algebra $\mathfrak{sl}(2,\mathbb{R})$ and  $\mathfrak{so}(3)$.
  
  Let us define the Lie system on $G^{os}$ of the form
  \begin{equation}\label{eq:Partial1}
  \frac{{\rm d}g}{{\rm d}t}=\sum_{\alpha=1}^rb_\alpha(t)X^R_\alpha(g)=:X^{os}(t,g)\,,\qquad g\in G^{os}\,,\qquad t\in\mathbb{R}\,,
  \end{equation}
  where $b_1(t),\ldots,b_6(t)$ are arbitrary $t$-dependent functions and $\{X^R_1,\ldots,X^R_6\}$ form a basis of the Lie algebra $V^R$ of right-invariant vector fields on a Lie group $G^{os}$ in such a way that $\{X^R_1,X_2^R,X_3^R\}$ forms a basis of the Lie algebra $V_{sl}:=\langle X_1^R,X_2^R,X_3^R\rangle\simeq\mathfrak{sl}(2,\mathbb{R})$ and $\{X^R_4,X^R_5,X_R^6\}$ stands for a basis of the Lie algebra $V_{so}:=\langle X_4^R,X_5^R,X_6^R\rangle\simeq \mathfrak{so}(3)$.
  
  The relevance of the Lie system (\ref{eq:Partial1}) is due to the fact that it appears in the solution of quantum harmonic oscillators with a spin-magnetic term \cite{Dissertationes}. More specifically, particular cases of (\ref{eq:Partial1}) appear as the automorphic Lie systems associated with the $t$-dependent Schr\"odinger equation related to a quantum $t$-dependent Hamiltonian operator of the form
  \begin{equation}\label{Eq:Op1}
  \widehat{H}_1(t)=\frac{\widehat{p}_x^2+\widehat{p}_y^2+\widehat{p}_z^2}{2}+\omega^2(t)\frac{\widehat{x}^2+\widehat{y}^2+\widehat{z}^2}{2}+B_x(t)\widehat{S}_x+B_y(t)\widehat{S}_y+B_z(t)\widehat{S}_z,
  \end{equation}
  where  $\omega(t)$ is any $t$-dependent frequency, $B_x(t),B_y(t),B_z(t)$ are the coordinates of a magnetic field and $\widehat{S}_x, \widehat{S}_y, \widehat{S}_z$ can be assumed to be, for instance, the Pauli matrices up to a proportional constant. Note that $\widehat{H}_1(t)$ acts on $\mathcal{L}^2(\mathbb{R}^3)\otimes \mathbb{C}^2$, where $\mathcal{L}^2(\mathbb{R}^3)$ is the space of complex square-integrable functions on $\mathbb{R}^3$. It is remarkable that (\ref{eq:Partial1}) is also related to other $t$-dependent Hamiltonian operators, e.g.
  \begin{equation}\label{Eq:Op2}
  \widehat{H}_2(t)=\frac{\widehat{p}_x^2+\widehat{p}_y^2+\widehat{p}_z^2}{2}+\omega^2(t)\frac{\widehat{x}^2+\widehat{y}^2+\widehat{z}^2}{2}+B_x(t)\widehat{L}_x+B_y(t)\widehat{L}_y+B_z(t)\widehat{L}_z,
  \end{equation}
  acting on $\mathcal{L}^2(\mathbb{R}^3)$, where $\widehat L_x,\widehat{L}_y,\widehat{L}_z$ are the angular momentum operators relative to the axes OX, OY, OZ, respectively. 
  In both previous cases, the key to understand the relation between (\ref{Eq:Op1}) and (\ref{Eq:Op2}) is the fact that ${\rm i}\widehat{H}_j(t)$, with $j=1,2$, can be written as a linear combination with $t$-dependent constants of a family of skew-Hermitian operators closing a real Lie algebra isomorphic to $\mathfrak{g}^{os}$ (cf. \cite{Dissertationes}). 
  
The above examples show that we can define a momentum map of the form
$$\textstyle
\mathcal{J}^{os}:G^{os}\times \mathcal{G}(\mathfrak{g}^{os}) \rightarrow \bigwedge^\bullet T^*G^{os}\,.
$$  
We call $\mathcal{G}(\mathfrak{g}^{os})$ a {\it traceless Grassmannian}, i.e. the space of traceless subalgebras of $\mathfrak{g}^{os}$.
This mapping allows us to gather all previous reduction processes.
  Although the relevant properties of the system are better understood by using a geometric approach, we hereafter provide a local coordinate expression for this system by using canonical coordinates of the second kind on $G^{os}$ of the form
  $$\phi:(v_1,\ldots,v_6)\in \mathfrak{g}^{os}\mapsto \exp(-v_1e_1)\times\ldots\times\exp(-v_6e_6)\in G^{os}\,,$$ 
  for $X_\alpha^R(e)=-e_\alpha$ for $\alpha=1,\ldots,6$,
  namely
  \begin{equation}\label{WNQH}
  \begin{dcases}
  \frac{dv_1}{dt}=b_1(t)+b_2(t)\, v_1+b_3(t)\,v_1^2\,,&\frac{dv_4}{dt}=b_4(t)+\tan v_5(b_5(t)\sin v_4 + b_6(t)\cos v_4)\,,\\
  \frac{dv_2}{dt}=b_2(t)+2\,b_3(t)\,v_1\,,&\frac{dv_5}{dt}=b_5(t)\cos v_4 - b_6(t)\sin v_4 \,,\\
  \frac{dv_3}{dt}=e^{v_2}\,b_3(t)\,, &\frac{dv_6}{dt}=\frac{b_5(t)\sin v_4 + b_6(t)\cos v_4}{\cos v_5}\,.
  \end{dcases}
  \end{equation}
We have now the following coordinate expressions for the vector fields
  \begin{equation}\label{Rel}
  \begin{aligned}
    X^R_1&=\frac{\partial}{\partial v_1}\,,&X^R_4&=\frac{\partial }{\partial v_4}\,,\\
    X^R_2&=v_1\frac{\partial}{\partial v_1}+\frac{\partial }{\partial v_2}\,,&X^R_5&=\sin v_4\tan v_5\frac{\partial}{\partial v_4} + \cos v_4\frac{\partial}{\partial v_5} + \frac{\sin v_4}{\cos v_5}\frac{\partial}{\partial v_6}\,,\\
    X^R_3&=v_1^2\frac{\partial }{\partial v_1}+2v_1\frac{\partial }{\partial v_2}+e^{v_2}\frac{\partial }{\partial v_3}\,,&X^R_6&=\cos v_4\tan v_5\frac{\partial}{\partial v_4} - \sin v_4\frac{\partial}{\partial v_5} + \frac{\cos v_4}{\cos v_5}\frac{\partial}{\partial v_6}\,.
  \end{aligned}
  \end{equation}
  Recall that right-invariant vector fields must be defined everywhere. The singularity in the above local coordinate expression when $v_5=\pi/2+k\pi$ for $k\in \mathbb{Z}$ is only due to the fact that  canonical coordinates of the second kind are defined, generically, only on an open neighbourhood of the neutral element of $G^{os}$. A simple change of coordinates would remove the singularity in the expression of right-invariant vector fields when the previous second canonical coordinates fail to be well defined.
  
  Next, it will be useful to describe the commutation relations among the above vector fields
  {\small
    \begin{equation*}
    \begin{aligned}
      & [X^R_1,X^R_2] = X^R_1\,, & [X^R_2,X^R_3] & = X^R_3\,, & [X^R_3,X^R_4] & = 0\,, & [X^R_4,X^R_5] & = X^R_6\,, & [X^R_5,X^R_6] & = X_4\,, \\
      & [X^R_1,X^R_3] = 2 X^R_2\,, & [X^R_2,X^R_4] & = 0\,, & [X^R_3,X^R_5] & = 0\,, & [X^R_4,X^R_6]& = -X_5^R\,, && \\
      & [X^R_1,X^R_4] = 0\,, & [X^R_2,X^R_5] & = 0\,, & [X^R_3,X^R_6] & = 0\,, &&&& \\
      & [X^R_1,X^R_5] = 0\,, & [X^R_2,X^R_6] & = 0\,, &&&&&& \\
      & [X^R_1,X^R_6] = 0\,.&&&&&&&&
    \end{aligned}
    \end{equation*}
  }
  
  Thus, the vector fields given in \eqref{Rel} span a real 6-dimensional Lie algebra $V^{os}$ isomorphic to $\mathfrak{sl}_2\oplus\mathfrak{so}_3$. Their dual forms are given by 
  \begin{equation*}
 \begin{aligned}
   & \eta_1^R = \d v_1 - v_1\d v_2 + v_1^2 e^{-v_2}\d v_3\,, & \eta_4^R & = \d v_4 - \sin v_5 \d v_6\,,\\
   & \eta_2^R = \d v_2 - 2v_1 e^{-v_2}\d v_3\,, & \eta_5^R & = \cos v_4\d v_5 + \sin v_4\cos v_5\d v_6\,,\\
   & \eta_3^R = e^{-v_2}\d v_3\,, & \eta_6^R & = -\sin v_4\d v_5 + \cos v_4\cos v_5\d v_6\,.
 \end{aligned}
 \end{equation*}
  
  Following the method presented in \cite{GLMV19}, we can build up a multisymplectic form on $G^{os}= SL_2\times SO_3$ such that the elements of $V^{os}$ become Hamiltonian vector fields with respect to it.
  Indeed, the obtained multisymplectic form reads
 \begin{equation}\label{eq:vol-form-os}
  \Theta^{os}=\eta_1^R\wedge \dotsb\wedge \eta^R_6=e^{-v_2}\cos v_5\,\d v_1\wedge\d v_2\wedge\d v_3\wedge\d v_4\wedge\d v_5\wedge\d v_6\,.
 \end{equation}
  
 We will use the Lie symmetries of the multisymplectic Lie system $(G^{os},\Theta^{os},X^{os})$ to reduce it. Recall that the Lie symmetries of a Lie system are given by the system of partial differential equations \eqref{Eq:SysRos}. In this case, the Lie algebra $V^S = {\rm Sym}(V^{os})$ is spanned by the vector fields
 \begin{equation}\label{eq:syms-CM}
   \begin{aligned}
     & Y_1^L = e^{v_2}\frac{\partial}{\partial v_1} + 2v_3\frac{\partial}{\partial v_2} + v_3^2\frac{\partial}{\partial v_3}\,, & Y_4^L & = \frac{\cos v_6}{\cos v_5}\frac{\partial}{\partial v_4} - \sin v_6\frac{\partial}{\partial v_5} + \cos v_6\tan v_5\frac{\partial}{\partial v_6}\,,\\
     & Y_2^L = \frac{\partial}{\partial v_2} + v_3\frac{\partial}{\partial v_3}\,, & Y_5^L & = \frac{\sin v_6}{\cos v_5}\frac{\partial}{\partial v_4} + \cos v_6\frac{\partial}{\partial v_5} + \sin v_6\tan v_5\frac{\partial}{\partial v_6}\,,\\
     & Y_3^L = \frac{\partial}{\partial v_3}\,, & Y_6^L & = \frac{\partial}{\partial v_6}\,,
   \end{aligned}
 \end{equation}
where the only non-vanishing commutation relations are
 \begin{equation*}
     \begin{aligned}
       & [Y_1^L,Y_2^L] = -Y_1^L\,, & [Y_2^L,Y_3^L] & = -Y_3^L\,, & [Y_4^L,Y_5^L] & = -Y_6^L\,, & [Y_5^L,Y_6^L] & = -Y_4^L\,,\\
       & [Y_1^L,Y_3^L] = -2Y_2^L\,, &&& [Y_4^L,Y_6^L] & = Y_5^L\,.
     \end{aligned}
 \end{equation*}

 Hence, $V^S = \langle Y_1^L,Y_2^L,Y_3^L,Y_4^L,Y_5^L,Y_6^L\rangle\simeq V^{os}$ via the Lie algebra isomorphism mapping $X_\alpha^R\mapsto -Y_\alpha^L$ for $\alpha=1,\ldots,6$. %Let $\mathfrak{g}_{CM}$ be the abstract Lie algebra isomorphic to $V^{os}$.

Let us consider the momentum map
$$
\begin{array}{rcccc}
\mathcal{J}^{os} \colon & G^{os}\times \mathcal{G}(\mathfrak{g}^{os}) &
\longrightarrow & \bigwedge^\bullet T^*G^{os}
\\
& (\xi,w) &
\longmapsto & \iota_{w}\Theta^{os}(\xi) \,.
\end{array}
$$

$\bullet$ First, consider the reduction of $\Theta^{os}$ by  $
Y_1^L\wedge Y_2^L\wedge Y_3^L$. We have the quotient map projection
$$ \pi_{123}\colon (v_1,\dotsc,v_6)\in G^{os}\longmapsto (v_4,v_5,v_6)\in G/SL_2\simeq SO_3\,. $$
The projections of the elements of the basis \eqref{Rel} of the Vessiot--Guldberg Lie algebra $V^{os}$ onto $G^{os}/SL_2$ read $\widetilde X_1^{sl} = \widetilde X_2^{sl} = \widetilde X_3^{sl} = 0$, while
\begin{align*}
    \widetilde X_4^{sl} &= \frac{\partial}{\partial v_4}\,,\\
    \widetilde X_5^{sl} &= \sin v_4\tan v_5\frac{\partial}{\partial v_4} + \cos v_4\frac{\partial}{\partial v_5} + \frac{\sin v_4}{\cos v_5}\frac{\partial}{\partial v_6}\,,\\
    \widetilde X_6^{sl} &= \cos v_4\tan v_5\frac{\partial}{\partial v_4} - \sin v_4\frac{\partial}{\partial v_5} + \frac{\cos v_4}{\cos v_5}\frac{\partial}{\partial v_6}\,.
\end{align*}
It is clear that $\widetilde V^{sl} = \langle\widetilde X_4^{sl},\widetilde X_{5}^{sl},\widetilde X_6^{sl}\rangle$ is a three-dimensional Lie algebra isomorphic to $\mathfrak{so}_3$. This gives rise to a Lie system
$$ X^{sl} = \sum_{\alpha = 4}^6 b_\alpha(t)\widetilde X_\alpha^{sl} $$
on the quotient space $G_1^{os} = G^{os}/SL_2$. This is a multisymplectic Lie system with respect to the multisymplectic three-form $\Theta^{os}_{123}$ determined by
$$ \pi_{123}^*\Theta^{os}_{123} = \iota_{Y_1^L\wedge Y_2^L\wedge Y_3^L}\Theta^{os} = \cos v_5 \,\d v_4\wedge\d v_5\wedge\d v_6 $$
on $G_1^{os}$. 

$\bullet$ On the other hand, consider the reduction of $\Theta^{os}$ by  $
Y_4^L\wedge Y_5^L\wedge Y_6^L$. We have the projection
$$ \pi_{456}\colon (v_1,\dotsc,v_6)\in G^{os}\longmapsto (v_1,v_2,v_3)\in G^{os}/SO_3\simeq SL_2\,. $$
The projections of the elements of the basis \eqref{Rel} of the Vessiot--Guldberg Lie algebra $V^{os}$ onto $G^{os}/SO_3$ read
$$ \widetilde X_1^{so} = \frac{\partial}{\partial v_1}\,,\qquad \widetilde X_2^{so} = v_1 \frac{\partial}{\partial v_1} + \frac{\partial}{\partial v_2}\,,\qquad \widetilde X_3^{so} = v_1^2\frac{\partial}{\partial v_1} + 2v_1\frac{\partial}{\partial v_2} + e^{v_2}\frac{\partial}{\partial v_3}\,, $$
and $\widetilde X_4^{so} = \widetilde X_5^{so} = \widetilde X_6^{so} = 0$. It is clear that $\widetilde V^{so} = \langle\widetilde X_1^{so},\widetilde X_2^{so},\widetilde X_3^{so}\rangle$ is three-dimensional Lie algebra isomorphic to $\mathfrak{sl}_2$. This gives rise to a Lie system
$$ X^{so} = \sum_{\alpha = 1}^3 b_\alpha(t)\widetilde X_\alpha^{so} $$
on the quotient space $G_2^{os} = G^{os}/SO_3$. This is a multisymplectic Lie system on $G_2^{os}$ with respect to the multisymplectic three-form $\Theta^{os}_{456}$ induced by
$$ \pi^*_{456}\Theta^{os}_{456} = \iota_{Y_4^L\wedge Y_5^L\wedge Y_6^L}\Theta^{os} =- e^{-v_2}\d v_1\wedge\d v_2\wedge\d v_3\,. $$

\end{subsubsection}

 %%%%%%%%%%%%%%%%%%%%%%%%%%%%%%%%%%%%%%%%%%%%%%%%%%%%%%%%%%%%%%%%

 \section{Multisymplectic reconstruction}
 \label{section:reconstruction}
 %%%%%%%%%%%%%%%%%%%%%%%%%%%%%%%%%%%%%%%%%%%%%%%%%%%%%%%%%%%%%%%%

 This section gives criteria for determining a multisymplectic Lie system $(M,\Theta,X)$ from several of its multisymplectic reductions. Let us start by giving a particular example that illustrates our more general techniques to be deployed later on.

Let us go back to the multisymplectic Lie system $(G^{os},\Theta^{os},X^{os})$ related to (\ref{eq:Partial1}) and associated with quantum harmonic oscillators with a spin-magnetic term. Recall that $(G^{os},\Theta^{os},X^{os})$ gives rise, by using multisymplectic reductions, to two different multisymplectic Lie systems 
 $$
(G^{os}_1, \Theta^{os}_{123}, X^{sl}) \,,\qquad (G^{os}_2, \Theta^{os}_{456}, X^{so}\,)\,.
 $$
  Recall that $\Theta^{os}_{123}$ is a multisymplectic form on $ G^{os}_1$ such that $\pi^*_{123}\Theta^{os}_{123}=\iota_{Y_1^L\wedge Y_2^L\wedge Y_3^L}\Theta^{os}$.  As
 \begin{equation}\label{Eq:RedMul}
Y^L_1\wedge Y^L_2\wedge Y^L_3= e^{v_2}\frac{\partial}{\partial v_1}\wedge\frac{\partial}{\partial v_2}\wedge\frac{\partial}{\partial v_3}\,, 
 \end{equation}
  and $\Theta^{os}$ is 6-nondegenerate because it is a volume form on $G^{os}$, 
the original multisymplectic form is uniquely determined by $\Theta^{os}_{123}$, namely $\Theta_{123}^{os}=\cos v_5\d v_4\wedge \d v_5\wedge \d v_6$, and it must take the form \eqref{eq:vol-form-os}. Alternatively, the knowledge of $\Theta^{os}_{456}$, namely $\Theta^{os}_{456}=-e^{-v_2}\d v_1\wedge d v_2\wedge \d v_3$, and the fact that it was obtained as a multisymplectic reduction via the contraction of $\Theta^{os}$ with
$$  Y^L_4\wedge Y^L_5\wedge Y^L_6= \frac{1}{\cos v_5}\frac{\partial}{\partial v_4}\wedge\frac{\partial}{\partial v_5}\wedge\frac{\partial}{\partial v_6}\,,
 $$
 determines again the value of $\Theta^{os}$, which matches the value obtained previously via the previous multisymplectic reduction. 
 
 Let us now focus on retrieving the $t$-dependent vector field $X^{os}$ on $G^{os}$, whose reductions gave rise to the $t$-dependent vector fields on $G^{os}_1$ and $G^{os}_2$, given by 
 \begin{equation}\label{Eq:Reduc}
X^{sl} = \sum_{\alpha = 4}^6 b_\alpha(t)\widetilde X_\alpha^{sl}\,,\qquad X^{so} = \sum_{\alpha = 1}^3 b_\alpha(t)\widetilde X_\alpha^{so}\,,
 \end{equation}
 respectively. In this case, we know that $\pi_{123*}X^{os}=X^{sl}$, which determines $X^{os}$ up to a vector field on $G^{os}$ taking values in $\ker T\pi_{123}$. Similarly, $\pi_{456*}X^{os}=X^{so}$, which retrieves the value of $X^{os}$ up to a vector field  on $G^{os}$ taking values in $\ker T\pi_{456}$. Hence, the value of $X^{os}$ is determined at a point $g\in G^{os}$ up to an element of 
 $$
 \ker T_p\pi_{123}\cap \ker T_p\pi_{456}=\{ 0\}\,.
 $$
In other words, $X^{os}$ is uniquely determined by (\ref{Eq:Reduc}). Note that this example shows how to determine the original multisymplectic Lie system on a Lie group, related to a multisymplectic volume form, via two of its multisymplectic reductions.

%\subsection{General theory}
The previous example illustrates how to set sufficient conditions to retrieve our initial multisymplectic Lie system from its multisymplectic reductions. For instance, we can consider the following proposition.

\begin{proposition}\label{Prop:Red1} Let $G$ be a Lie group with certain Lie subgroups $G_1,\ldots, G_k$. Let $(G/G_i,\Theta^i,X^{i})$, with $i=1,\ldots,k$, be a set of multisymplectic reductions of a multisymplectic Lie system  $(G,\Theta,X)$, let $\pi_i:G\rightarrow G/G_i$ and $\iota_{Y_i}\Theta$, with $1,\ldots,k$, be the projections and Noether invariants related to the multisymplectic reductions induced by the multivector fields $Y_1,\ldots, Y_k$ on~$G$. Assume that 
$$
\bigcap_{i=1}^k\ker T_g\pi_i=\{ 0\}\,,\qquad \forall g\in G\,,
$$
and suppose that $\Theta$ is a volume form on~$G$. Then, the multisymplectic Lie system $(G,\Theta,X)$ is univocally determined from its multisymplectic reductions.
\end{proposition}
\begin{proof} Since by construction of the multisymplectic reductions, one has that $X^i_{\pi_i(g)}=T_g\pi_iX_g$ for $i=1,\ldots,k$ and every $g\in G$, then the value of $X$ at $g$, let us say $X_g$, is  determined by $X^1,\ldots,X^k$ up to an element of
$
\bigcap_{i=1}^k\ker T_g\pi_i\,,
$
which is, by assumption, equals to zero. Hence, $X$ is determined univocally by $X^1,\ldots,X^k$. 

Let us now study the determination of the multisymplectic form $\Theta$ on $G$ via the Noether invariants $\iota_{Y_i}\Theta$ associated with the multisymplectic reductions and the induced multisymplectic forms on the different quotients $G/G_i$. Given one of them, let us say $\Theta^i$, one has that 
$$
\iota_{Y_i}\Theta=\pi^*_i\Theta^i\,.
$$
Since $\Theta$ is a volume form and $\pi$ is a surjective submersion, $\Theta$ is univocally determined by $\Theta^i$, the multivector field $Y_i$, and the above relation.
\end{proof}

To understand the generality of the assumptions of the previous proposition, it is worth noting that multisymplectic volume forms are not as restrictive as it may seem at first. In two-dimensional manifolds, every multisymplectic form is a volume form. On three-dimensional manifolds, it occurs the same way as differential two- and one-forms are degenerate. In 4-dimensional manifolds, differential three-forms are degenerate because they can all be written, locally, as the contraction of a volume form with a vector field. Moreover, differential one-forms are also degenerate and the only multisymplectic forms on a 4-dimensional manifold are volume forms or symplectic ones. It is worth stressing that we are mainly interested in multisymplectic forms that are not symplectic ones, as only in this case the multisymplectic formalism will provide information that is not available by means of other theories. Hence, non-volume multisymplectic forms, which are not symplectic forms, appear only on manifolds of dimension five or higher. 

Anyhow, Proposition \ref{Prop:Red1} can be modified to deal with multisymplectic Lie systems relative to a class of non-volume multisymplectic forms. As shown next, these multisymplectic forms need not be $k$-nondegenerate for $k>1$, and this leads to new features. Let us provide an example to illustrate the new geometric aspects introduced in this case.

Consider a multisymplectic form on $\mathbb{R}^8$ of the form
$$
\Omega^S=\!\!\!\!\!\!\!\!\sum_{1\leq i_1<i_2<i_3<i_4<i_5<i_6\leq 8}\!\!\!\!\!\!\!\!\d x^{i_1}\wedge \d x^{i_2}\wedge \d x^{i_3}\wedge \d x^{i_4}\wedge \d x^{i_5}\wedge \d x^{i_6},
$$
where $\{x^1,\ldots,x^8\}$ are global coordinates on~$\mathbb{R}^8$. In fact, $\Omega^S$ is closed and, since the mappings $\Omega^{S\sharp} _x: T_xM\rightarrow \bigwedge^5 T_xM$ are injections for every $x\in M$, one-nondegenerate. Moreover, $\Omega^S$ is also two- and three-nondegenerate, i.e. if $\iota_Z\Omega^S=0$ for  $Z\in \mathfrak{X}^2(\mathbb{R}^8)$ or $Z\in \mathfrak{X}^3(\mathbb{R}^8)$, then $Z=0$. In fact, $\iota_Z\Omega^S=0$ can be seen as an algebraic equation in the coefficients of $Z\in \mathfrak{X}^s(\mathbb{R}^8)$ in a certain basis of differential $s$-forms. In particular, when $Z\in \mathfrak{X}^2(\mathbb{R}^8)$, the algebraic system $\iota_Z\Omega^S=0$ can be seen as a system of $\displaystyle\binom{8}{4}$ linear equations on the $\displaystyle\binom{8}{2}$ coefficients of $Z$, which, after a simple but tedious calculation, turns out to have a unique trivial $Z=0$ solution. Moreover, if $Z\in \mathfrak{X}^3(\mathbb{R}^8)$, the algebraic system $\iota_Z\Omega^S=0$ can be seen as a system of  $\displaystyle\binom{8}{3}$ linear equations on the $\displaystyle\binom{8}{3}$ coefficients of 
$Z$, which, as in the previous case, has only a unique zero trivial solution. Note that this last fact also implies that $\Omega^S$ must be one- and two-nondegenerate.

Consider the Lie system on $\mathbb{R}^8$ of the form
$$
X^S=\sum_{i=1}^8b_i(t)\frac{\partial}{\partial x^i},
$$
where $b_1(t),\ldots,b_8(t)$ are arbitrary $t$-dependent functions. It is immediate that the vector fields $\{\partial/\partial x^1,\ldots,\partial/\partial x^8\}$ are Hamiltonian relative to $\Omega^S$. Then, $(\mathbb{R}^8,\Omega^S,X^S)$ becomes a multisymplectic Lie system.

Consider the Lie group action $\Phi:(\lambda^1,\ldots,\lambda^8,x^1,\ldots,x^8)\in \mathbb{R}^8\times \mathbb{R}^8\mapsto (\lambda^1+x^1,\ldots,\lambda^8+x^8)\in \mathbb{R}^8$, which admits a momentum map
$$\textstyle
\mathcal{J}: \mathbb{R}^8\times  \mathcal{G}(\mathbb{R}^{8})   \rightarrow \bigwedge^\bullet T^*\mathbb{R}^{8}\,.
$$
Given the Hamiltonian bivector fields on $\mathbb{R}^8$ of the form
$$
Z_a=\frac{\partial}{\partial x^1}\wedge \frac{\partial}{\partial x^2}\,,\qquad Z_b=\frac{\partial}{\partial x^4}\wedge \frac{\partial}{\partial x^5}\,,\qquad  Z_c=\frac{\partial}{\partial x^7}\wedge \frac{\partial}{\partial x^8}\,,
$$
let us analyse the  multisymplectic reductions of $\Omega^S$ related to them. In particular,
$$
\iota_{Z_a}\Omega^S=\!\!\!\!\sum_{3\leq i_1<i_2<i_3<i_4\leq 8}\!\!\!\!\!\!\d x^{i_1}\wedge \d x^{i_2}\wedge \d x^{i_3}\wedge \d x^{i_4}\,,\qquad 
 \iota_{Z_c}\Omega^S=\!\!\!\!\sum_{1\leq i_1<i_2<i_3<i_4\leq 6}\!\!\!\!\!\!\d x^{i_1}\wedge \d x^{i_2}\wedge \d x^{i_3}\wedge \d x^{i_4}\,,
$$
$$
\iota_{Z_b}\Omega^S=\!\!\!\!\sum_{1\leq i_1<i_2<i_3<i_4\leq 8}\!\!\!\!\!\!\epsilon_{45i_1i_2i_3i_4}\d x^{i_1}\wedge \d x^{i_2}\wedge \d x^{i_3}\wedge \d x^{i_4}\,,
$$ 
where $\epsilon_{45i_1i_2i_3i_4}$ is the Levi--Civita symbol relative to the sequence of indices $4,5,i_1,i_2,i_3,i_4$.
The above Noether invariants are basic forms relative to the canonical projections $\pi_a:(x^1,\ldots,x^8)\in \mathbb{R}^8\mapsto (x^3,\ldots,x^8)\in \mathbb{R}^6$, $\pi_b:(x^1,\ldots,x^8)\in \mathbb{R}^8\mapsto (x^1,x^2,x^3,x^6,x^7,x^8)\in \mathbb{R}^6$, and $\pi_c:(x^1,\ldots,x^8)\in \mathbb{R}^8\mapsto (x^1,x^2,x^3,x^4,x^5,x^6)\in \mathbb{R}^6$ associated with the multisymplectic reductions of $\Omega^S$ by $Z_a, Z_b$, and $Z_c$, respectively.  In other words, $\pi_a^*\Omega_a=\iota_{Z_a}\Omega^S$, $\pi_b^*\Omega_b=\iota_{Z_b}\Omega^S$, $\pi_c^*\Omega_c=\iota_{Z_c}\Omega^S$ for three unique differential 4-forms $\Omega_a, \Omega_b$, and $\Omega_c$ on~$\mathbb{R}^6$, which are therefore closed. They are also one-nondegenerate because $\Omega^S$ is three-nondegenerate. 

On the other hand, the value of $\pi^*_a\Omega_a$ does not allow us to characterise univocally $\Omega^S$   via the equality $\iota_{Z_a}\Omega^S=\pi^*_a\Omega_a$   since there are many differential 6-forms, $\Omega$, satisfying that $\iota_{Z_a}\Omega=0$. Therefore, $\pi^*_a\Omega_a=\iota_{Z_a}\Omega^S$ determines $\Omega^S$ up to a differential 6-form, let us say $\Omega_{aa}$, whose contraction with $Z_a$ vanishes. Moreover, the differential 6-forms $\Omega^S+\epsilon\Omega_{aa}$, for small $\epsilon$, are one-nondegenerate and $\iota_{Z_a}(\Omega^S+\epsilon\Omega_{aa})=\pi^*_a\Omega_a$. Hence, to determine every possible multisymplectic 6-form on $\mathbb{R}^8$, let us say $\Omega$, satisfying $\iota_{Z_a}\Omega=\pi^*_a\Omega_a$, it is interesting to define the following notion.

\begin{definition} Let $w$ be a $r$-tangent vector at $x\in M$, i.e. $w\in \bigwedge^rT_xM$. The {\it $\ell$-annihilator} of $w$,  denoted by $w^{\circ,\ell}$, is the space of skew-symmetric $\ell$-forms, $\Omega_x$, on $T_xM$ such that $\iota_w\Omega_x=0$.
\end{definition}

In the previous example, the knowledge of $\pi^*_a\Omega_a$ determines the value of $\Omega^S$ up to a differential 6-form vanishing on~$Z_a$. Meanwhile, 
$$
Z_a^{\circ,6}\ni \d x^{i_1}\wedge \d x^{i_2}\wedge \d x^{i_3}\wedge \d x^{i_4}\wedge \d x^{i_5}\wedge \d x^{i_6}
$$
if, and only if, $\{1,2\}\not\subset\{i_1,\ldots,i_6\}$, which determines $Z_a^{\circ,6}$. Similarly, one can determine $Z_b^{\circ,6}$ and $Z_c^{\circ,6}$. From this, it follows that $Z_a^{\circ,6}\cap Z_b^{\circ,6}\cap Z_c^{\circ,6}=\{0\}$ and $\Omega^S$ is the unique differential 6-form satisfying that $\iota_{Z_a}\Omega^S=\pi_a^*\Omega_a$, $\iota_{Z_b}\Omega^S=\pi_b^*\Omega_b$, and  $\iota_{Z_c}\Omega^S=\pi_c^*\Omega_c$. Hence, it is immediate that Proposition \ref{Prop:Red1} can be generalised as follows.

\begin{proposition}%\label{Prop:Red1}
Let $G$ be a finite-dimensional Lie group with certain Lie subgroups $G_1,\ldots, G_k$. Let $(G/G_i,\Theta^i,X^{i})$, with $i=1,\ldots,k$, be a set of multisymplectic reductions of a multisymplectic Lie system  $(G,\Theta,X)$;
where $\Theta$ is an $(s+1)$-nondegenerate multisymplectic $\ell$-form on~$G$. 
Let $\pi_i:G\rightarrow G/G_i$ and $\Theta^i=\iota_{Z_i}\Theta$, with $1,\ldots,k$ and $Z_i\in \mathfrak{X}^s(G)$, be the projections and Noether invariants related to the multisymplectic reductions. If
$$
\bigcap_{i=1}^k\ker T_g\pi_i=\{ 0\}\quad\text{and}\quad\bigcap_{i=1}^k Z_i^{\circ,\ell}(g)=\{0\}\,,\qquad \forall g\in G\,,
$$
then, the multisymplectic Lie system $(G,\Theta,X)$ is univocally determined from its multisymplectic reductions $(G/G_i,\Theta^i,X^i)$, with $i=1,\ldots,k$.
\end{proposition}

%%%%%%%%%%%%%%%%%%%%%%%%%%%%%%%%%%%%%%%%%%%%%%%%%%%%%%%%%%%%%%%%
\section{Conclusions and outlook}
%%%%%%%%%%%%%%%%%%%%%%%%%%%%%%%%%%%%%%%%%%%%%%%%%%%%%%%%%%%%%%%%

This work presents a reduction procedure for multisymplectic Lie systems, which were first introduced in \cite{GLMV19}, by using a suitable momentum map. We have also developed a reconstruction procedure which, in favourable cases, allows us to recover the original Lie system from its reductions. These results are illustrated by working out several examples, including the Schwarz equation, dissipative quantum harmonic oscillators, a control system, and a quantum harmonic oscillator with a spin-magnetic term.

In the future, it is planned to study new physical and mathematical applications of our methods, e.g.\ to Schwarz derivatives \cite{LS20}. It would also be interesting to investigate multisymplectic Lie systems that are not of locally transitive type. In those cases, the Lie system is not locally diffeomorphic to an automorphic Lie system on a Lie group and we cannot use the structures on Lie groups to obtain multisymplectic forms, Lie symmetries, reductions and so on. This also raises the question of the existence of a multisymplectic form that turns a Lie system into a multisymplectic one. 

Note that shown results about Schwarz equations highlight that multisymplectic structures may allow for the study of relative equilibrium points for systems that are not Hamiltonian ones relative to any symplectic or Poisson structure. It would be interesting to study the extension of energy-momentum methods \cite{MS88} to Hamiltonian systems relative to multisymplectic forms. The energy-momentum method is based heavily on a Marsden--Weinstein  reduction procedure (cf. \cite{LZ21,MS88}), which has been depicted here for multisymplectic structures. Therefore, our work solves one of the main drawbacks in developing such a method.  We hope to follow these lines of research in upcoming works.

Finally, it would be interesting to develop an analogue of our methods to study the reductions of Lie systems admitting a Vessiot--Guldberg Lie algebra of Hamiltonian vector fields relative to other geometric structures, e.g.\ $k$-cosymplectic \cite{dLeon2015} or $k$-contact \cite{GGMRR20} structures.

%%%%%%%%%%%%%%%%%%%%%%%%%%%%%%%%%%%%%%%%%%%%%%%%%%%%%%%%%%%%%%%%
\section*{Acknowledgements}
%%%%%%%%%%%%%%%%%%%%%%%%%%%%%%%%%%%%%%%%%%%%%%%%%%%%%%%%%%%%%%%%

The authors acknowledge fruitful discussions and comments from our colleague Miguel-C. Mu\~noz-Lecanda.
We also acknowledge partial financial support from the 
project MINIATURA 5 of the Polish National Science Centre (NCN) under grant number Nr 2021/05/X/ST1/01797,
and also from the
Spanish Ministerio de Ciencia e Innovaci\'on projects
PGC2018-098265-B-C31, 
PGC2018-098265-B-C33
and
RED2018-102541-T.
We also thank the anonymous referees, whose interesting comments and suggestions have helped us improve our article.

\addcontentsline{toc}{section}{References}
%%%%%%%%%%%%%%%%%%%%%%%%%%%%%%%%%%%%%%%%%%%%%%%%%%%%%%%%%%%%%%%%

%%%%%%%%%%%%%%%%%%%%%%%%%%%%%%%%%%%%%%%%%%%%%%%%%%%%%%%%%%%%%%%%
\end{document}